\newif\iflncs
\lncsfalse
\iflncs
\documentclass[runningheads]{llncs}
\usepackage[T1]{fontenc}
\else
\documentclass[12pt]{article}
\fi

\iflncs
\usepackage{amssymb}
\else
\usepackage{style}
\fi

\usepackage{physics}
\usepackage[inline]{enumitem}
\usepackage{array}
\usepackage{multicol}
\usepackage{cleveref}
\usepackage{comment}

\usepackage{thmtools}
\usepackage{thm-restate}

\usepackage{tikz}
\usetikzlibrary{quantikz2}

\iflncs

\newcommand\GM[1]{}
\newcommand\tamer[1]{}

\newcommand{\poly}{\mathrm{poly}}
\newcommand{\polylog}{\mathrm{polylog}}
\declaretheorem[name=Protocol]{protocol}

\newcommand{\bin}{\{0,1\}}
\newcommand{\C}{\mathsf{C}}

\else

\renewcommand{\bin}{\{0,1\}}
\renewcommand{\C}{\mathsf{C}}

\newcommand\GM[1]{\textcolor{red}{[G: #1]}}
\newcommand\tamer[1]{\textcolor{magenta}{[T: #1]}}

\declaretheorem[style=definition,name=Protocol]{protocol}
\fi

\crefname{protocol}{Protocol}{Protocols}
\crefname{claim}{Claim}{Claims}

\newcommand{\FF}{\mathbb{F}}

\newcommand{\NN}{\mathbb{N}}
\newcommand{\Expct}{\mathbb{E}}

\renewcommand{\sec}{\lambda}
\newcommand{\adv}{\mathcal{A}}

\newcommand{\Bad}{\mathrm{Bad}}

\newcommand{\X}{\mathsf{X}}
\newcommand{\U}{\mathsf{U}}
\newcommand{\V}{\mathsf{V}}

\newcommand{\Y}{\mathsf{Y}}
\newcommand{\B}{\mathsf{B}}

\newcommand{\TD}{\mathbf{TD}}

\newcommand{\padv}{\mathcal{P}^*}
\newcommand{\badv}{\mathcal{B}}

\newcommand{\equal}{\texttt{equal}}
\newcommand{\accept}{\texttt{accept}}

\renewcommand{\epsilon}{\varepsilon}

\title{How to Verify that a Small Device is Quantum, Unconditionally}
\date{}

\iflncs
\author{Anonymous Submission}
\institute{}
\else
\author{Giulio Malavolta\thanks{Bocconi University. {\tt giulio.malavolta@unibocconi.it}. Work supported by the European Research Council through an ERC Starting Grant (Grant agreement No.~101077455, ObfusQation) and partially funded by the Deutsche Forschungsgemeinschaft (DFG, German Research Foundation) under Germany's Excellence Strategy - EXC 2092 CASA – 390781972.} \and Tamer Mour\thanks{Bocconi University. {\tt tamer.mour@unibocconi.it}. Work supported by European Research Council (ERC) under the EU’s Horizon 2020 research and innovation programme (Grant agreement No. 101019547).}}
\fi
\begin{document}
\maketitle

\begin{abstract}
    A proof of quantumness (PoQ) allows a classical verifier to efficiently test if a quantum machine is performing a computation that is infeasible for any classical machine. In this work, we propose a new approach for constructing PoQ protocols where soundness holds \emph{unconditionally} assuming a bound on the \emph{memory} of the prover, but otherwise no restrictions on its runtime. In this model, we propose two protocols:
\begin{itemize}
    \item A simple protocol with a quadratic gap between the memory required by the honest parties and the memory bound of the adversary. The soundness of this protocol relies on Raz's (classical) memory lower bound for matrix inversion (Raz, FOCS 2016).
    \item A protocol that achieves an exponential gap, building on techniques from the literature on the bounded storage model (Dodis et al., Eurocrypt 2023).
\end{itemize}
Both protocols are also efficiently verifiable. Despite having worse asymptotics, our first protocol is conceptually simple and relies only on arithmetic modulo $2$, which can be implemented with one-qubit Hadamard and CNOT gates, plus a \emph{single} one-qubit non-Clifford gate. 
\end{abstract}

\iflncs
\else
\newpage
\tableofcontents
\newpage
\fi

\section{Introduction}

Quantum computing promises to bring the computing sciences to a new world, where the laws of quantum mechanics are harnessed to do computation. The first step towards this ambitious goal is to obtain an experimental demonstration of \emph{quantum advantage}: The construction of a quantum apparatus capable of performing computations that are manifestly infeasible for any classical machine.  Existing approaches combine techniques from cryptography, complexity theory, and physics, in order to design a \emph{proof of quantumness} (PoQ) that is (i) efficient enough to be run on existing quantum machines and (ii) \emph{efficiently verifiable} by a fully classical computer. Unfortunately, despite a massive industrial and academic effort devoted to achieving this milestone, a convincing experimental demonstration remains elusive.

To set some context, let us discuss existing approaches for the experimental demonstration of quantum advantage. On the one hand, the \emph{circuit sampling} approach asks the quantum machine to (approximately) sample from a distribution defined by a quantum circuit, that is conjectured to be hard to simulate classically. The advantage of this approach is that it is feasible for near-term quantum machines, and is in fact the approach adopted by Google's first experiment \cite{supremacy}, and more recently in \cite{willow}. On the downside, the classical hardness of sampling from these distributions has been called into question \cite{closinggap}. Moreover, by far the biggest problem of this approach is that a classical machine cannot efficiently verify that the computation was done correctly.

To overcome this drawback, researchers have turned their attention to \emph{cryptography}. The simplest proposal that allows efficient verification would be to run Shor's algorithm \cite{FOCS:Shor94} to factor a large number. Then, given its prime decomposition, it is easy to check that the computation was done correctly. More sophisticated and general approaches exist, where the classical hardness is proven assuming the intractability of computational problems related to lattices \cite{FOCS:BCMVV18}, group actions \cite{TCC:AlaMalRah22}, or discrete logarithms \cite{dlog}. More recently, proposals based on compiling two-player non-local games \cite{STOC:KLVY23} have emerged as an alternative pathway to PoQ. Common to all of these approaches is that the (classical) verification is always efficient, making them an appealing alternative to the circuit sampling approach. On the flip-side, their soundness also relies on unproven computational conjectures and, furthermore, the quantum operations are somewhat more complex than the ones involved in the sampling method, making this approach less suitable for experiments. To the best of our knowledge, none of the existing experiments \cite{experimental} was performed in a regime where the problem is classically hard.

% citations in https://arxiv.org/pdf/2203.15877

To summarize, all known PoQ protocols rely on computational assumptions, and are either inefficiently verifiable or beyond the current technological reach. We emphasize that the former limitation is for a good reason: It was recently shown that PoQ protocols \emph{imply} the existence of (quantum) computationally hard problems \cite{MSY24}. In this work, we propose a new approach for testing quantumness. Instead of placing assumptions on the computational power of the prover, we constrain its \emph{memory}.

\subsection{Our Contributions}

We consider a model where the adversary has a limited amount of memory, but is otherwise computationally unbounded. This is commonly referred to as the \emph{bounded storage model} \cite{Maurer92,CM97,Ding01,JC:Vadhan04,DHRS04,EC:GuaZha19,DQW23}.
We show that in this model, one can construct protocols that are (i) efficiently verifiable and (ii) unconditionally sound.

Our first result is a simple PoQ protocol with a quadratic gap between the space complexity of the honest execution and the memory bound on the malicious prover.\footnote{For the moment, it suffices to think of the space complexity of a quantum prover as the logarithm of the dimension of his Hilbert space. We make this more precise as the discussion progresses.} The protocol is extremely simple; the honest quantum prover performs simple arithmetics modulo $2$, and can be implemented with one-qubit Hadamard and CNOT gates and a \emph{single} one-qubit non-Clifford gate. The prover's circuit is illustrated in \Cref{fig:qart}. The protocol is inspired by techniques from the cryptographic literature \cite{CVQC,BCMVV21,BGKPV23,NZ23,BK24,BKMSW24}, and its soundness relies on Raz's (classical) memory lower bound for learning parities \cite{Raz}. We summarize our result with the following statement.

\begin{theorem}[Simple PoQ with Quadratic Gap, Informal]
    There exists a proof of quantumness protocol with constant gap between completeness and soundness, where the honest execution runs $O(n^2)$ time, the verifier has $O(n)$ memory, and the honest prover has $n+2$ qubits. Soundness holds unconditionally against any classical attacker that uses less than $n^2/20$ bits of memory.
\end{theorem}
%
%\tamer{NEW:} 
We stress that, although the verifier's runtime and the prover's quantum circuit size are quadratic, an honest execution of the protocol has linear parallel runtime. In particular, the quantum circuit implementing the prover has depth $O(n)$.

Despite its simplicity, an unsatisfactory aspect of the above protocol is that it has only a quadratic gap between the memory of the honest and the malicious prover. Leveraging techniques from the literature of protocols in the bounded storage model \cite{DQW23}, we show that this gap can be made arbitrarily large. The following theorem establishes the theoretical feasibility of a PoQ with an exponential gap.

\begin{theorem}[PoQ with Exponential Gap, Informal]
    There exists a proof of quantumness protocol with constant gap between completeness and soundness, where the honest execution execution runs $\poly(n)$ time, the verifier has $\polylog (n)$ memory, and the honest prover has $\polylog (n)$ qubits. Soundness holds unconditionally against any classical attacker that uses less than $n$ bits of memory.
\end{theorem}
As our final contribution, we show that the above protocol achieves a form of soundness also against \emph{quantum provers}. In more details, it allows the generation of a state of the form
\[
\frac{\ket{x_0} + \ket{x_1}}{\sqrt{2}}
\]
such that even the prover that holds that state in memory cannot guess both $x_0$ and $x_1$. In the literature, this is known as a \emph{claw state generation} protocol \cite{BK24}.

\begin{theorem}[Claw Generation with Quantum Soundness, Informal]\label{thm:claw-gen-quantum-informal}
    There exists a claw generation protocol, 
    where the honest execution runs $\poly(n)$ time, the verifier has $\polylog (n)$ memory, and the honest prover has $\polylog (n)$ qubits. Soundness holds unconditionally against any quantum attacker that uses less than $n$ qubits.
\end{theorem}
%
%\tamer{NEW. If you think is unnecessary, feel free to remove. I think it can add to the intrigue of the story although it doesn't concern the main result} 

Claw generation protocols have proven to be extremely valuable in quantum cryptography. For instance, all known protocols to classically verify BQP computation rely on claw states, one way or another \cite{CVQC,BCMVV21,BGKPV23,NZ23,BK24,BKMSW24}. To demonstrate the usefulness of our claw-generation protocol, we sketch how to derive analogous PoQ protocols with unconditional security in the bounded storage model, by adapting analysis from prior work~\cite{BK24} (\cref{sec:appendix}).

% We believe that the bounded-storage model is of particular relevance in practice when considering security against quantum adversaries, e.g., in potential applications of \Cref{thm:claw-gen-quantum-informal}. It is reasonable to assume that at some point along the foreseeable progress in building quantum computers, small quantum computers that can compute on, say, 100 qubits will exist while larger computers of size even quadratic, i.e. 10k qubits, will be still considered completely infeasible.

\subsection{Technical Overview}

Our starting point is a recent result of Raz~\cite{Raz}, who considers the following experiment between a verifier and a prover:
\begin{enumerate}
        \item The verifier samples a uniform $s \gets \mathbb{F}_2^n$.
        \item For $i=1\dots m$, the verifier samples a uniform $v_i \gets\mathbb{F}_2^n$ and sends $(v_i, v_i^\intercal s)$ to the prover.
        \item The prover returns some $\Tilde{s}\in\FF^n_2$.
\end{enumerate}
The main theorem of~\cite{Raz} shows that, if a prover uses less than $n^2/20$ memory bits, then the probability that he outputs $s$ is $2^{-\Omega(n)}$, for any $m = \poly(n)$. This implies a strong lower bound on the memory needed to solve linear equation systems, for any classical computer. One could speculate that quantum algorithms are better suited for solving this kind of tasks, thus immediately yielding a PoQ protocol based on the memory-hardness of solving linear equations. Alas, we are not aware of any quantum algorithm solving this problem with $O(n)$ qubits, bringing us back to square one.

Nevertheless, inspired by a recent work from Guan and Zhandry~\cite{EC:GuaZha19}, we observe that the above problem has sufficient \emph{structure} to allow us to leverage techniques developed in the context of computationally security. Turning our attention to cryptographic PoQs, we see that all recent works~\cite{BCMVV21,KCVY22,KLVY23,BGKPV23} follow the same three-phase template:
\begin{itemize}
    \item (Claw Generation) In the first phase, the prover and verifier engage in an interaction, at the end of which the internal state of the (honest) prover is of the form
    \[
    \frac{\ket{x_0}+\ket{x_1}}{\sqrt{2}},
    \]
    where $x_0$ and $x_1$ are two bitstrings that are known to the verifier. Furthermore, it is required that the prover cannot output both $x_0$ and $x_1$ simultaneously. Thus, this phase relies on cryptographic hardness and is typically realized using a \emph{trapdoor claw-free function} (TCF)~\cite{Mah18FHE}. We refer to the above state as a \emph{claw state}, and to this interaction as a \emph{claw generation} subroutine.
    \item (Commitment) In the second phase, the prover is instructed to convert the claw state into a single qubit $b\in\{\ket{0},\ket{1},\ket{+},\ket{-}\}$. The security guarantee from the claw state generation phase translates to the fact that the prover cannot tell whether $b$ is in the $Z$-basis or in the $X$-basis, as this would imply learning the values $x_0$ and $x_1$.
    \item (Non-Local Game) The protocol then completes with a \emph{CHSH test}: A test of quantumness that is derived by the CHSH non-local game, where a quantum prover is challenged to produce a correlation with the ``hidden basis'' that a classical prover cannot.
\end{itemize}
Note that after the first phase, all guarantees are information-theoretic. In other words, cryptography is used only to derive the hardness of recovering the values in the claw state. Thus, our task to design a PoQ in the bounded storage model reduces to the task of designing a claw state generation protocol with security against memory-bounded provers.

\paragraph{A Simple PoQ from Learning Parities.} The above experiment considered by Raz already suggests a construction for a claw state generation protocol. Consider the linear map $x\mapsto Ax$, where $A=(V,Vs)\in\FF^{m\times(n+1)}$ is induced by samples in the experiment, namely composed by rows of the form $(v_i,v^\intercal_i s)$. We make two simple observations:
\begin{enumerate}
    \item $A$ is two-to-one with overwhelming probability, since $V$ is full rank with overwhelming probability.
    \item Finding a \emph{claw}, namely $x_0,x_1$ such that $Ax_0=Ax_1$, implies finding $s$ since $x_0+x_1=(s,-1)$ is the only non-trivial element in $\ker(A)$. Computing $s$ given samples of the form $(v_i,v_i^\intercal s)\in\FF^{n+1}$ is precisely the problem of learning parities considered in the experiment, and cannot be done with less than $n^2/20$ memory.
\end{enumerate}

Hence, $x\mapsto Ax$ is a \emph{claw-free function} in the bounded storage model, namely a 2-to-1 function where it is hard to find a claw. Claw generation from such a function is straight-forward: the verifier sends the function to the prover, the prover computes the function over the uniform superposition then measures an outcome $y$, reducing the state to a superposition of the two pre-images of $y$.

Naively performing the above, however, requires the parties to communicate $A$ and, in particular, to remember it using quadratic space. Instead, we perform the evaluation in rounds: at rounds $i=1,\dots,m$ of the interaction, the verifier samples a row from $A$, i.e. a vector $a_i=(v_i,v_i^\intercal s)\in\FF^{n+1}_2$, and sends it to the prover. The prover starts with a uniform superposition over $n+1$ qubits in a register $\X$ and, at round $i$, applies the isometry
\[
\ket{x}_\X\mapsto\ket{x}_\X\ket{a^\intercal_i x}_\Y,
\]
measures $\Y$ and returns the outcome $y_i\in\FF_2$ to the verifier. Overall, such an interaction corresponds to the prover computing the mapping $\ket{x}_\X\mapsto \ket{x}_\X\ket{Ax}_{\Y}$, then measuring the $m$ qubits in $\Y$ to obtain an outcome $y=(y_1,\dots,y_m)$ and a residual claw state $(\ket{x_0}+\ket{x_1})/\sqrt{2}$, where $x_0,x_1$ are the two pre-images of $y$ under $A$.

% However, we are not done: The naive method to generate a claw state would be to evaluate the linear map coherently and measure its output, which would require $O(n^2)$ memory also from the honest prover. Fortunately, it turns out that there is a better method, that allows an honest prover can compute the claw state using only $O(n)$ qubits.

% To describe the idea, let us go back to the interactive experiment. First, 

While the above protocol seems to give us the claw generation sufficient for a PoQ, we have omitted an important detail: How can the verifier compute $x_0$ and $x_1$, which are generally necessary to carry out the CHSH test? Once again, computing $x_0$ and $x_1$ given $y$ requires the verifier to invert the matrix $A$. Since $A$ has quadratic size, it is infeasible for the verifier in our model to remember it, let alone invert it. Taking a closer look at the literature, we observe that existing instantiations of this outline, e.g.,~\cite{KCVY22}, do not require the verifier to extract $x_0$ and $x_1$ entirely but rather only the bits $b_0=r^\intercal x_0$ and $b_1=r^\intercal x_1$ for a uniformly random $r\gets\FF^{n+1}_2$ \emph{sampled by the verifier} in the commitment phase, independently of $x_0$ and $x_1$. In fact, it is required that the verifier can tell if $b_0$ and $b_1$ are equal -- in which case $r^\intercal(s,-1)=r^\intercal(x_0+x_1)=0$ -- and, only when they are, to know their value.

Given this observation, we propose a memory-efficient implementation of the verifier. At the beginning of the protocol, the verifier flips a coin to decide whether $r$ satisfies $r^\intercal(s,-1)=0$ or $r^\intercal(s,-1)=1$. In the former case, the verifier is not required to learn $b_0$ and $b_1$ and, therefore, the CHSH test can be executed with a uniformly random $r\notin\ker((s,-1))$, independently of $x_0,x_1$. In the latter case, the verifier must know $r^\intercal x_0$. We use the fact that the prover communicates $y$, which is a linear function of $x_0$, to the verifier. Instead of sampling $r\gets\ker((s,-1))$ directly, the verifier samples $r$ as a linear combination of the rows of $A$: when he sends $a_i$ at round $i$, he decides ``on-the-fly'' whether to include it in the linear combination or not, with probably $1/2$ each. Remembering his choices, the verifier completes the claw generation with a linear combination $u\in\FF^m_2$ satisfying $u^\intercal A=r$. In particular, it holds that $r^\intercal x_0=u^\intercal A x_0=u^\intercal y$. This the verifier can compute given $u$ and $y$ in linear space, allowing him to execute the CHSH test also in this case.

\paragraph{A PoQ with an Exponential Gap.} Our first PoQ, while extremely simple, achieves only a quadratic gap between the space complexity of the honest parties and the memory needed for a successful attacker. A natural question is whether we can achieve a better gap.

We positively answer this question in our second result. We devise an interactive claw generation protocol that is secure against attackers with memory $m$ while requiring only $\polylog m$ memory for an honest execution. The protocol is based on \emph{interactive hashing}~\cite{NOVY,DHRS04}. In prior work~\cite{CCM98,Ding01,DQW23}, (classical) interactive hashing was used to construct oblivious transfer (OT) protocols in the BSM. Most relevant to us is the work by Dodis et al.~\cite{DQW23} where, roughly speaking, interactive hashing is used to let a sender and a receiver jointly choose two bit values from a long stream $u_1,\dots,u_k\in\bin$, such that the verifier knows both values but a bounded-memory receiver can remember only one of them by the end of the protocol. We observe that realizing this outline with a \emph{quantum} receiver, using a coherent implementation of interactive hashing~\cite{MY23}, allows the receiver to obtain a superposition of the two chosen bits while still preventing him from recovering both simultaneously. %This already sounds very much like claw generation.

Let us recall the outline from~\cite{DQW23} in more detail. A sender sends a stream of $k\gg m$ uniform bits $u_1,\dots,u_k\gets\bin$, one at a time, to the receiver. The sender remembers the bits at two random locations $v^*_0,v^*_1\gets[k]$ of his choice. After the streaming is complete, the two parties perform interactive hashing. The transcript of the interactive hashing protocol defines a 2-to-1 hash function $h$ over the domain $[k]$ and a hash value $y$, and guarantees the following: \begin{enumerate*}[label=(\roman*)]
    \item\label{item:ih-completeness} On the one hand, the receiver can control one pre-image of $y$ under $h$ and, in fact, he can make it so $y=h(v)$ for an apriori arbitrarily chosen input $v$. \item \label{item:ih-soundness} On the other hand, the receiver cannot control \emph{both} pre-images of $y$ under $h$. Namely, for any bounded-size set of inputs $B$, apriori chosen by the receiver, it holds that $h^{-1}(y)\subseteq B$ with a very small probability. 
\end{enumerate*}
Consequently, the interactive hashing protocol defines a pair of indices $v_0,v_1\in[k]$, and the two chosen bits are set to be $u_{v_0}$ and $u_{v_1}$. By the soundness of the interactive hashing (\ref{item:ih-soundness}), the receiver cannot obtain both $u_{v_0}$ and $u_{v_1}$ by the end of the protocol. This is shown in~\cite{DQW23} as follows: Define $B$ to be the set of bits in the stream about which the receiver remembers sufficient information. Via standard incompressibility argument, due to the bounded memory of the receiver, $B$ cannot be too large. Hence, the receiver cannot have too much information about both $u_{v_0}$ and $u_{v_1}$.

We carry the above outline where the verifier plays the role of the sender and the quantum prover plays the role of the classical receiver, coherently. That is, the prover prepares a uniform superposition over $[k]$, next to an ancilla qubit. In the streaming phase, upon the receipt of the $v^{th}$ bit $u_v$, the prover maps the basis vector $\ket{v}\ket{0}$ to $\ket{v}\ket{u_v}$, entangling the two registers. Consequently, the streaming completes with the prover's internal state being
$$
\frac{1}{\sqrt{k}}\sum_{v\in[k]}\ket{v}_\V\ket{u_v}_\U.
$$

Next, the verifier and the prover, with input register $\V$, perform interactive hashing to select two values $v_0,v_1\in[k]$. The residual state of the prover after this is complete is
$$
    \frac{\ket{v_0,u_{v_0}}+\ket{v_1,u_{v_1}}}{\sqrt{2}},
$$
where $v_0,v_1$ are the two pre-images of $y$ under $h$, for $y$ and $h$ defined by the transcript of the interactive hashing (recall in a classical invocation, it holds by \ref{item:ih-completeness} that $h(v)=y$ for the prover's classical input $v$). The verifier then checks that $h(v^*_0)=h(v^*_1)=y$. If not, the protocol is repeated with a new random stream of bits. Otherwise, the parties have performed successful claw generation: the prover has a superposition of two values that the verifier knows, yet of which he can remember at most one. By extending the analysis from~\cite{DQW23} to the setting of a memory-bounded \emph{quantum} adversary, we additionally show that such claw generation guarantees security also against quantum adversaries. The analysis turns out to be quite technical. In particular, we rely on a lemma from~\cite{BBK22} that can be seen as an analog of incompressibility arguments for quantum information. For more details we refer the reader to \Cref{sec:qhardness}.

While the above claw generation already provides non-trivial hardness for computing the values in the claw, such hardness is limited to that of predicting a \emph{single bit} (note we do not claim that the prover cannot obtain the indices $v_0,v_1$). Naturally, such a ``1-bit claw generation'' cannot provide too much security. We propose a simple security amplification strategy via \emph{stitching}. In stitching, we convert many ``1-bit claws'' into a single claw that is as hard to predict as predicting all of the 1-bit claws that compose it, simultaneously.

Let us describe how to stitch two claws together. Stitching more claws generalizes straight-forwardly. Let the prover's state be $(\ket{v^1_0,u^1_0}+\ket{v^1_1,u^1_1})\otimes (\ket{v^2_0,u^2_0}+\ket{v^2_1,u^2_1})$ after the generation of two claw states. The verifier, who has the values $v^1_0,v^1_1,v^2_0,v^2_1$, sends to the prover two predicates $g^1,g^2:[k]\to\bin$ such that $g^i(v^i_b)=b$. Stitching is complete with the prover applying the following isometry over the registers containing the $v$ values
$$
    \ket{v^1}_{\V^1}\ket{v^2}_{\V^2}\mapsto \ket{v^1}_{\V^1}\ket{v^2}_{\V^2}\ket{g^1(v^1)\oplus g^2(v^2)}_{\B},
$$
then measuring the qubit in $\B$ to obtain a bit $b\in\bin$, which he sends to the verifier. This operation entangles the two claws, resulting in the following stitched claw state
$$
    \frac{\ket{v^1_0,u^1_0}\ket{v^2_b,u^2_b}+\ket{v^1_1,u^1_1}\ket{v^2_{1-b},u^2_{1-b}}}{\sqrt{2}},
$$
which the verifier can anticipate given $b$.

To conclude, our claw generation protocol performs the above 1-bit claw generation sequentially $\sec$ times, where $\sec$ is a security parameter, then performs stitching to stitch together the $\sec$ 1-bit claws into a claw that a memory-bounded attacker can break with only negligible probability.

\section{Preliminaries}

We denote by $[n]$ the set $\{1, \dots, n\}$. We recall the following fundamental fact about the rank of random binary matrices (see, e.g.~\cite{BKW97}).
\begin{proposition}[Rank of a Random Matrix]\label{prop:rank}
    Let $M \gets \mathbb{F}_2^{m \times n}$ be a uniformly sampled random matrix, with $m = 2n$. Then,
    \[
    \Pr\left(\rank(M) = n \right) = 1 - O(2^{-n}).
    \]
\end{proposition}

\subsection{Memory Bounded Algorithms}

We say that an adversary is \emph{memory bounded} if the size of its state at any point in the computation is bounded by some parameter $m$. In fact, for most of our bounds (for instance the one in \cref{sec:exp}) we can also consider a slightly stronger adversary that is allowed to have an unlimited amount of short-term memory, but can only store an $m$-bit state after the end of each round. 
%\textcolor{red}{GM: Double check the above} \tamer{i think it's true up to the bound of Raz, where i don't know what happens}

\begin{lemma}[\cite{Raz}]\label{lmm:raz}
    For any $C < 1/20$ there exists $\alpha > 0$, such that the following holds. For $m\leq 2^{\alpha n}$ and an algorithm $\mathcal{A}$, consider the experiment:
    \begin{itemize}
        \item Sample a uniform $s \in \mathbb{F}_2^n$.
        \item For $i=1\dots m$: Sample a uniform $v_i \in\mathbb{F}_2^n$ and send $(v_i, v_i\cdot s)$ to $\mathcal{A}$.
        \item $\mathcal{A}$ returns some $\Tilde{s}$.
    \end{itemize}
    If $\mathcal{A}$ uses less than $C n^2$ memory bits, then:
    \[
    \Pr\left(s = \Tilde{s}\right)\leq O(2^{-\alpha n}).
    \]
\end{lemma}

\subsection{The Goldreich-Levin Extractor}

We recall the well-known Goldreich-Levin extractor~\cite{GL89}, highlighting its memory complexity.

\begin{lemma}[\cite{GL89}]\label{lmm:GL}
    Let $f:\mathbb{F}_2^n \to \mathbb{F}_2$ be a function such that, for some $x\in\FF^n_2$,
    \[
    \Pr_{r \in \mathbb{F}_2^n}\left(f(r) = r^\intercal x\right) \geq 1/2 + \varepsilon(n).
    \]
    for inverse-polynomial $\varepsilon(n)$. Then there exists an extractor with memory $O(n \log n)$ running in time polynomial in $n$ and making oracle calls to $f$, that returns $x$ with probability inverse-polynomial in $n$.
\end{lemma}
\begin{proof}[Proof Sketch]
    We briefly recall the description of the extractor. For $t = O(\log n/\varepsilon^2(n)) = O(\log n)$ the extractor proceeds as follows:
    \begin{enumerate}
        \item Sample uniformly random $(r_1, \dots, r_t) \gets \mathbb{F}_2^n$ and $(b_1, \dots, b_t) \gets \mathbb{F}_2$.
        \item For $i = 1\dots n$,
        \begin{enumerate}[label*=\arabic*.]
            \item For any $S \subseteq [t]$, compute
            \[
            x_{i, S} = f\left( \sum_{j\in S}r_j\right) +\sum_{j\in S}b_j.
            \]
            \item Set $x_i = \mathrm{maj}\left\{x_{i,S}\right\}_S$.
        \end{enumerate}
        \item Return $(x_1, \dots, x_n)$.
    \end{enumerate}
    It is clear that the runtime of the algorithm is polynomial in $n$ and it is shown in~\cite{GL89} that the above extractor succeeds with at least inverse-polynomial probability in $n$. As for the memory complexity, the algorithm must store $(r_1, \dots, r_t)$ and $(b_1, \dots, b_t)$ and, for each bit, compute the majority function which is computable in constant memory due to Barrington's theorem~\cite{Barrington}. Overall, the memory required by the extractor is bounded by $O(n\log n)$.
\end{proof}

\subsection{Information Theory}

The \emph{Shannon entropy} (or simply entropy) of a random variable $X$ is defined as $H(X)=\Expct_{x\gets X}(-\log \Pr_{X}(X=x))$. The conditional entropy of $X$ given another random variable $Y$ is defined as $H(X\mid Y)=\Expct_{y\gets Y}(H(X\mid Y=y))$.

The \emph{min-entropy} of a random variable $X$ is defined as $H_\infty(X)=-\log\max_x \Pr(X=x)$. We define the \emph{conditional min-entropy} of $X$ given another random variable $Y$ as: \[H_\infty(X\mid Y)=-\log \Expct_{y\gets Y}(\max_x\Pr(X=x\mid Y=y)) = -\log \Expct_{y\gets Y}(2^{-H_\infty(X\mid Y=y)}).\] Note that $H_\infty(X\mid Y)\leq H(X\mid Y)$ for any $X,Y$ (in particular this holds for the non-conditional notions).

We recall the following basic facts on min-entropy from the literature.
\begin{lemma}[\cite{DORS08}]\label{lem:min-entropy-conditional}
    For any random variables $X,Y,Z$, where $Y$ is over $\bin^m$, we have $H_\infty(X\mid Y,Z)\geq H_\infty(X\mid Z) - m$.
\end{lemma}

\begin{lemma}[\cite{DORS08}]\label{lem:min-entropy-tail-bound}
    For any random variables $X,Y$ and any $\varepsilon > 0$, it holds that $$ \Pr_{y\gets Y}(H_\infty(X\mid Y=y) \geq H_\infty(X\mid Y)-\log(1/\varepsilon))\geq 1-\varepsilon. $$
\end{lemma}

\begin{proposition}[\cite{DQW23}]\label{prop:independent-condition}
    For random variables $X,Y,Z$ where $X$ and $Y$ are independent conditioned on $Z$, it holds that $H_\infty(X\mid Y)\geq H_\infty(X\mid Y,Z) \geq H_\infty(X\mid Z)$ and $H_\infty(X,Y\mid Z)\geq H_\infty(X|Z)+H_\infty(Y|Z)$.
\end{proposition}

\begin{proposition}[\cite{DQW23}]\label{prop:min-entropy-block-entropy}
    For any random variables $X=X_1,\dots,X_k$ and $Y$, it holds that $$H_\infty(X\mid Y)\leq \sum_{i\in[k]}H(X_i\mid Y).$$
\end{proposition}
\begin{proof}
    The inequality follows easily from the chain-rule of Shannon entropy $$H_\infty(X\mid Y)\leq H(X\mid Y)=\sum_{i\in[k]}H(X_i\mid Y,X_1,\dots,X_{i-1})\leq \sum_{i\in[k]} H(X_i\mid Y).$$
\end{proof}

In the following proposition, we give a lower bound on the min-entropy of a binary random variable with high Shannon entropy.

\begin{proposition}[Binary Entropy Bound]\label{prop:entropy-to-min-entropy}
    For any binary random variable $X$, letting $H(X)=h$, it holds that $$H_\infty(X)\geq 1-\log(1+\sqrt{1-h^{\ln 4}}).$$
\end{proposition}
\begin{proof}
    Let $p=\max_{b\in\bin}\Pr(X=b)=2^{-H_\infty(X)}>1/2$. We start from the following known upper bound on the binary entropy function~\cite{Topsoe01}: $H(X)\leq (4p(1-p))^{1/\ln 4}$. The bound implies $4p-4p^2-h^{\ln 4}\geq 0$ and, consequently, $p\leq \frac{1}{2}(1+\sqrt{1-h^{\ln 4}})$.
\end{proof}
We recall some useful concentration bounds.
\begin{lemma}[Markov Inequality]\label{lmm:markov}
Let ${X}$ be a non-negative random variable and let $\alpha > 0$, then:
\[
\Pr\left(X \geq \alpha\cdot \Expct(X)\right) \leq 1/\alpha.
\]
\end{lemma}
\begin{lemma}[Chernoff Inequality]\label{lmm:chernoff}
Let ${X}_1 \dots {X}_n$ be independent random variables such that ${X}_i \in\{0,1\}$ and let $\Tilde{X} = \Expct(\sum_{i}{X}_i)$. Then:
\[
\Pr\left(\sum_{i}{X}_i \geq (1+\delta)\cdot \Tilde{X}\right)\leq e^{-\frac{\Tilde{X}\delta^2}{3}}.
\]
\end{lemma}
We say that a set of random variables ${X}_1 \dots {X}_n$ is \emph{negatively correlated} if for every subset $S \subseteq [n]$ it holds that:
\[
\Pr\left(\prod_{i\in S} X_i = 1\right) \leq \prod_{i\in S}\Pr\left( X_i = 1\right).
\]
In~\cite{PS97}, it is shown that the Chernoff bound (\cref{lmm:chernoff}) continues to hold also for negatively correlated random variables.

\subsection{Quantum Information}

In quantum mechanics, physical systems are identified with Hilbert spaces~$\mathcal{H}$, and the states of the system are identified with positive semidefinite operators (PSD)~$\rho$ with unit trace, called \emph{density operators}. A state is called \emph{pure} if the density operator has rank one, and otherwise it is called \emph{mixed}. Any unit vector~$\ket \psi \in \mathcal{H}$ determines a pure state by the formula~$\rho = \ketbra{\psi}{\psi}$, and conversely any pure state can be written in this way. We often associate a Hilbert space with a register $\mathsf{X}$, and we denote a state in such register as $\ket{\psi}_{\mathsf{X}}$.

A \emph{measurement} with a finite outcome set~$\mathcal{O}$ is described by a collection of bounded operators $\{M_x\}_{x\in \mathcal{O}}$ acting on~$\mathcal{H}$ such that~$\sum_{x\in \mathcal{O}} M_x=Id$, referred to as a POVM. A \emph{quantum circuit} is a unitary operator that operates on $\mathcal{H}$ and is given by the composition of unitary gates (taken from some fixed universal gate set). The size of a quantum circuit is the number of gates used in that circuit. The qubits are typically split into input qubits and ancillas, which are assumed to be initialized in the $\ket{0}$. In its most general form, any physically-admissible quantum operation is described by a \emph{completely-positive trace-preserving} (CPTP) map from linear operators $\mathrm{L}(\mathsf{X})$ on a register $\mathsf{X}$ to linear operators $\mathrm{L}(\mathsf{Y})$ on a register $\mathsf{Y}$.
The \emph{trace distance} between two states $\rho$ and $\sigma$ is defined as:
\[
\TD(\rho, \sigma) = \frac{1}{2}\norm{\rho-\sigma}_1 = \frac{1}{2}\mathsf{Tr}\left(\sqrt{(\rho-\sigma)^\dagger(\rho-\sigma)}\right)
\]
where $\mathsf{Tr}$ is the trace of a matrix. The operational meaning of the trace distance is that 
$1/2\cdot  (1+\TD(\rho, \sigma))$ is the maximal probability that two states $\rho$ and $\sigma$ can be distinguished by any (possibly unbounded) quantum channel or algorithm.

Lastly, we recall a useful lemma from \cite{BBK22}.

% \textcolor{red}{Standard background on quantum states / measurements / circuits, a selection of \url{https://arxiv.org/pdf/1704.07309} -- see which notions we have used in the end. Define CTPT maps and the linear operators $\text{L}$.}

\begin{lemma}[Plug-In Lemma~\cite{BBK22}]\label{lem:plug-in}
	Let $X=X_1,\dots,X_k$ and $W$ be arbitrarily dependent random variables, where $W$ is over $m$ qubits. Then, it holds that:
	$$\TD\left((i,X_{<i},X_i,W),(i,X_{<i},X'_i,W)\right)\leq\sqrt{m/2k},$$
	where $i$ is uniformly random over $[k]$ and $X'_i$ is sampled according to the marginal distribution of $X_i$ given $X_{<i}$ (and is otherwise independent in $W$).
\end{lemma}

\section{Simple Proof of Quantumness with Quadratic Gap}\label{sec:simple}

The first proof of quantumness we present in this work is based on the space lower bound for learning with parities by Raz~\cite{Raz} (\Cref{lmm:raz}) and, consequently, achieves a quadratic gap between the space complexity of an honest execution and the space complexity of the best attack.

\begin{theorem}[PoQ with Quadratic Gap]\label{thm:simple-poq}
    There exists a protocol between a classical verifier and a quantum prover where the two parties take as input a security parameter $n\in\NN$ and, at the end of the interaction, the verifier either accepts or rejects, and
    \begin{itemize}
        \item (Completeness) The verifier accepts with probability $\cos^2(\pi/8)-O(2^{-n})$ when interacting with an honest prover.
        \item (Soundness) Any non-uniform classical prover $\padv$ that on input $n$ uses less than $n^2/20$ memory bits has success probability at most $3/4+O(2^{-n})$ in making the verifier accept.
        \item (Complexity) The verifier and (honest) prover run in time $O(n^2)$ and in space $O(n)$. In particular, the prover uses $n+2$ qubits.
    \end{itemize}
\end{theorem}

\subsection{The Protocol}

Our protocol follows the general framework of basing PoQ protocol on \emph{claw generation}~\cite{BCMVV21,KCVY22,BGKPV23}. We perform claw generation using on a simple linear function $x\mapsto Ax$, where $A=(V,Vs)\in\FF^{m\times(n+1)}$ for a random $V\gets\FF^{m\times n}$ and $s\gets\FF^n$. The function is 2-to-1 with overwhelming probability and it is claw-free for any adversary who has space at most quadratic due to the space complexity of learning parities by~\cite{Raz} (\Cref{lmm:raz}).

While PoQ based on claw generation generally requires the verifier to extract the values in the claw, for the above function this demands inverting $A$ (and, in particular, remembering it) and therefore considered infeasible in our bounded-storage model. We show how to nevertheless implement the verifier using only linear space. Due to the simplicity of our claw-free function, the prover in our protocol is not only space-efficient, but can be implemented by a very simple quantum circuit, which we illustrate in \cref{fig:qart}.

\begin{protocol}[Parity-based PoQ]\label{protocol:simple}
    Let $n$ be the security paramter of the protocol and let $m = 2n$. The protocol consists of an interaction between a quantum prover and a classical verifier, described as follows.
\begin{itemize}
    \item (Claw Generation)
    \begin{enumerate}
        \item The verifier samples a uniform $s \gets \mathbb{F}_2^n$ and a uniform $u=(u_1,\dots,u_m) \gets \mathbb{F}_2^m$ at the beginning of the protocol. The verifier sets $t = (s, -1)$ and $r = 0^{n+1}$.
    
        The prover prepares the uniform superposition:
        \[
        \ket{\psi} = \frac{1}{\sqrt{2^{n+1}}}\sum_{x\in\mathbb{F}_2^{n+1}}\ket{x} \in \mathbb{C}^{2^{n+1}}.
        \]
        \item For $i = 1 \dots m$,
        \begin{enumerate}[label*=\arabic*.]
            \item \label{step:a} The verifier samples a uniform $v_i \gets \mathbb{F}_2^n$ and sends to the prover $(v_i, v_i^\intercal s) = a_i$. The verifier updates $r = r + u_i\cdot a_i$.
            \item \label{step:y} The prover applies the following isometric mapping to his state
            \[
            \ket{x}_\X \mapsto \ket{x}_{\X}\ket{a^\intercal_i x}_{\Y}
            \]
            and measures the qubit in $\Y$ in the computational basis to obtain a bit $y_i\in\mathbb{F}_2$ and sends it to the verifier.
        \end{enumerate}
    \end{enumerate}
    \item (Commitment) 
    \begin{enumerate}
        \item\label{step:simple-r} The verifier samples a random bit $c \gets\{0,1\}$, then:
        \begin{itemize}
            \item If $c = 0$, sends $r$ to the prover.
            \item If $c=1$, samples a fresh uniform $r\gets\FF_2^{n+1}$ conditioned on $r^\intercal t= 1$ and sends it to the prover.
        \end{itemize}
        \item The prover applies the following isometric mapping to his state
        \[
            \ket{x}_\X \mapsto \ket{x}_{\X}\ket{r^\intercal x}_{\B},
        \]
        then measures $\X$ in Hadamard basis to obtain $d \in \mathbb{F}_2^{n+1}$ and sends it to the verifier.
    \end{enumerate}

    \item (CHSH Test)
    \begin{enumerate}
        \item The verifier samples an angle $\theta \in \{\pi/ 8, - \pi/8\}$ and sends $\theta$ to the prover.
        \item The prover measures the qubit in $\B$ in the basis
        \[
        \left\{\cos(\theta)\ket{0} + \sin(\theta)\ket{1}, \cos(\theta)\ket{0} - \sin(\theta)\ket{1}\right\}
        \]
        to obtain an outcome $b\in \{0,1\}$ and sends it to the verifier.
        \item The verifier accepts if the following conditions are satisfied:
        \begin{itemize}
            \item If $c = 0$, accept iff $u^\intercal y = b$, where $y = (y_1, \dots, y_m)$.
            \item Otherwise: \begin{tabular}[t]{l}
                 If $\theta = \pi/8$, accept iff $d^\intercal t = b$. \\
                 If $\theta = -\pi/8$,  accept iff $d^\intercal t \neq b$.
            \end{tabular}
        \end{itemize}
    \end{enumerate}
\end{itemize}
\end{protocol}

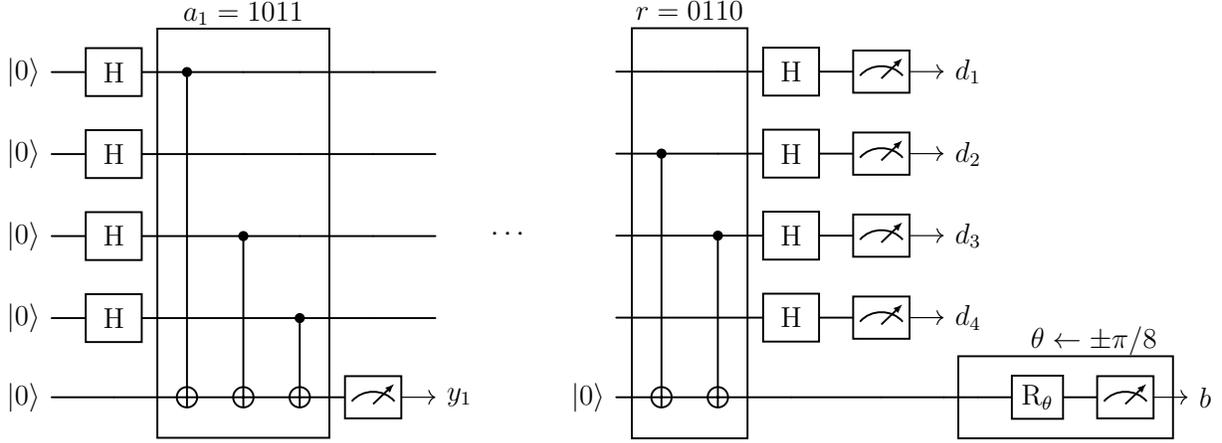
\begin{figure}
    \centering
\resizebox{\textwidth}{!}{

\begin{quantikz}
     \lstick{$\ket{0}$} & \gate{\text{H}}[] &  \ctrl{4} \gategroup[5,steps=3,style={inner
        sep=4pt}]{$a_1 = 1011$} & & & &\\
     \lstick{$\ket{0}$} & \gate{\text{H}}[] &  & & & &\\ 
     \lstick{$\ket{0}$} & \gate{\text{H}}[] &  & \ctrl{2} & & &\\ 
     \lstick{$\ket{0}$} & \gate{\text{H}}[] &  & & \ctrl{1} & &\\ 
     \lstick{$\ket{0}$} &  & \targ{} & \targ{} &  \targ{} & \meter{}\arrow[r] & \rstick{$y_1$}\setwiretype{n} &
    \end{quantikz}
    $\cdots\quad$
    \begin{quantikz}[slice style=blue]
     &  \gategroup[5,steps=2,style={inner
        sep=4pt}]{$r = 0110$} & &\gate{\text{H}}[]& \meter{}\arrow[r]&\rstick{$d_1$} \setwiretype{n} &&&\\
    & \ctrl{3} & &\gate{\text{H}}[] &\meter{}\arrow[r] &\rstick{$d_2$}\setwiretype{n}& &&\\ 
     &  & \ctrl{2} &\gate{\text{H}}[]& \meter{}\arrow[r]&\rstick{$d_3$}\setwiretype{n}& &&\\ 
     &  & &\gate{\text{H}}[]&  \meter{}\arrow[r] &\rstick{$d_4$}\setwiretype{n} &&&\\ 
     \lstick{$\ket{0}$} & \targ{} & \targ{} & & & &
      %\gategroup[1,steps=3,style={inner sep=4pt}]{$\text{ch} = \theta$}
    \gategroup[1,steps=3,style={inner sep=4pt}]{\quad\quad$\theta\gets\pm\pi/8$}
      &\gate{\text{R}_\theta}& \meter{} \arrow[r] &\rstick{$b$}\setwiretype{n}
    \end{quantikz}
}
    \caption{The quantum circuit executed by the prover in a run of the protocol with $n=4$. In each of the first $2n$ rounds the prover receives $a_i\in\FF^n_2$ from the verifier and responds by $y_i\in\FF_2$. In round $2n+1$ he receives $r\in\FF^n_2$ and returns $d\in\FF^n_2$. In the last round, he receives an angle $\theta\gets \{\pi/8,-\pi/8\}$ and responds by a bit $b$.}
    \label{fig:qart}
\end{figure}

\subsection{Analysis}

First, notice that the honest parties only require storage linear in $n$ to run the protocol. Specifically, the verifier only needs to keep in memory the variables $s$, $a_i$ (one at a time), $u$, $d$ and $y$, all of size $O(n)$. On the other hand, the prover's memory consists of $n+1$ classical bits (to store each $a_i$), along with $n+1$ qubits in $\X$ and an additional qubit in $\Y$ and $\B$ (one at a time).

%We next use the framework we set in \Cref{sec:template} to analyze the completeness and soundness of the protocol. While, at a first glance, \Cref{protocol:simple} does not perfectly fit in the template defined in \Cref{protocol:template}, we show that it is essentially a memory-efficient implementation of a protocol that does.

We start with showing that the protocol is complete, namely that an honest quantum prover is able to convince the verifier with high probability.

\begin{lemma}[Completeness]\label{lem:completeness-simple}
    The verifier in an honest execution of the protocol accepts with probability $\cos^2(\pi/8)-O(2^{-n})$. 
\end{lemma}

\begin{proof}
Let $\ell=n+1$ and let $A\in\FF^{m\times \ell}$ denote the matrix that is composed by the vectors $a_i$, sent by the verifier in Step~\ref{step:a} in the claw generation of \Cref{protocol:simple}, as its rows. Let $V\in\FF^{m\times n}$ denote the matrix similarly composed by the random vectors $v_i$. It holds that $A=(V,Vs)$. By \Cref{prop:rank}, $V$ has full rank with probability at least $1-O(2^{-n})$ and, therefore, $A$ defines a 2-to-1 linear map with such a probability.

We show that the verifier accepts with probability $\cos^2(\pi/8)$ conditioned on $A$ is 2-to-1. By the above, this is sufficient to derive the lemma.

It is convenient to delay the measurements done by the prover in Step~\ref{step:y} until after all iterations in claw generation are complete. This results in an equivalent residual state by the principle of delayed measurement. Applying the isometric map defined by $a_i$ to the prover state results into
    \[
        \frac{1}{\sqrt{2^{n+1}}}\sum_{x\in\mathbb{F}_2^{n+1}}\ket{x, \psi_x}_{\X\Y_{1\dots i-1}} \mapsto \frac{1}{\sqrt{2^{n+1}}}\sum_{x\in\mathbb{F}_2^{n+1}}\ket{x, \psi_x, a_i^\intercal x}_{\X\Y_{1\dots i}}.
    \]
    We can then rewrite the state of the prover in the end of claw generation (but before any measurement) as
    \[
    \frac{1}{\sqrt{2^{n+1}}}\sum_{x\in\mathbb{F}_2^{n+1}}\ket{x, A x}_{\X\Y_{1\dots m}}.
    \]
    Measuring the last $m$ qubits in registers $\Y_1,\dots,\Y_m$ in the computational basis, the prover obtains a $y \in \mathbb{F}_2^m$ and the state collapses to
    \[
    \frac{1}{\sqrt{2}}\sum_{x\in\mathbb{F}_2^{n+1} ~:~ Ax = y}\ket{x} = \frac{1}{\sqrt{2}}\left(\ket{x_0}+\ket{x_1}\right).
    \]
    Let $r\in\mathbb{F}_2^{\ell}$ be the string sent by the verifier in the commitment phase. The prover maps its state to
    \begin{equation}\label{eq:claw-state}
        \frac{1}{\sqrt{2}}\left(\ket{x_0}+\ket{x_1}\right)_\X \mapsto \frac{1}{\sqrt{2}}\left(\ket{x_0,  r^\intercal x_0}+\ket{x_1, r^\intercal x_1}\right)_{\X\B},
    \end{equation}
    then measures $\X$ in the Hadamard basis and gets $d\in \mathbb{F}_2^{\ell}$. We distinguish between two cases.
    \begin{itemize}
        \item (Case I) $c=0\implies r^\intercal x_0 = r^\intercal x_1=\delta$ for some $\delta\in\FF_2$: In such a case, the prover's state before measuring $\B$ (RHS of \Cref{eq:claw-state}) is a tensor product between a claw state in $\X$ and the qubit in $\B$, therefore, the measurement of $\X$ has no effect on $\B$. Thus, the prover's residual state after measuring is $\ket{\delta}$. In the CHSH test, the verifier verifies $b$ against the most likely outcome of the prover's measurement, which is $\delta=r^\intercal x_0=u^\intercal y$. It follows, then, that the prover succeeds with probability exactly $\cos^2(\pi/8)$.
        \item (Case II) $c=1\implies r^\intercal x_0 \neq r^\intercal x_1$: The residual state of the prover after measuring $\X$ is
        \[
            \frac{1}{\sqrt{2}}\left(\ket{0}+(-1)^{d\cdot(x_0 + x_1)}\ket{1}\right).
        \]
        Thus, if $\theta = \pi/8$, the most likely outcome of the prover's measurement is $d^\intercal (x_0+x_1)$, whereas if $\theta = -\pi/8$, the most likely outcome is $1-d^\intercal (x_0+x_1)$. Notice that if $A$ has rank $n$ then $\ker(A)$ has dimension 1, namely a single non-trivial element. Therefore, it must be the case that $x_0+x_1=t$. We conclude that the honest prover succeeds with probability exactly $\cos^2(\pi/8)$ in this case as well.
    \end{itemize}
\end{proof}

It remains to show that \Cref{protocol:simple} is sound against classical cheating provers. In the following lemma, we show that if a classical prover is able to convince the verifier with probability significantly better than $3/4$, then he is able to find the vector $s$ sampled by the verifier at the first step of claw generation. The hardness of finding $s$ is implied directly by the bound on space complexity of learning parities from \Cref{lmm:raz}; computing $s$ correspond precisely to learning parities give the verifier messages $(v_i,v_i^\intercal s)$ at claw generation.

Hence, together with \Cref{lem:completeness-simple}, the following lemma completes the proof of \Cref{thm:simple-poq}.

\begin{lemma}[Soundness]\label{lem:simple-soundness}
    Assume there exists a classical adversary $\padv(1^\sec)$ that makes the verifier from \Cref{protocol:simple} accept with probability larger than $3/4 + \epsilon(\sec)$ for some function $\epsilon=1/\poly$. Then, there exists an adversary $\adv(1^\sec)$ that interacts with the verifier in the claw generation sub-protocol (the first phase) and, at the end of the interaction, outputs $s$ with probability at least $1/\poly(\sec)$, where $s$ is the vector sampled by the verifier at the beginning of the protocol. Further, the space complexity of $\adv$ is larger than the space complexity of $\padv$ by an additive factor of at most $O(\ell\log\ell)$.
\end{lemma}
\begin{proof}
    First, we switch to a modified experiment where the matrix $A$ defined by the rows $a_i$ sent by the verifier in claw generation is a rank-$n$ matrix. Since this occurs with probability all but $O(2^{-n})$, the assumption on $\padv$ that it convinces the verifier holds also in the modified experiment.

    %Recall that, when $A$ has rank $n$, then $t$ is the only non-trivial element in its kernel. In particular, it holds that $x_0+x_1=t$ for any $x_0,x_1$ such that $Ax_0=Ax_1$. Hence, we may consider the goal of $\adv$ to be computing $t$.

    Let $W$ be the internal state of the prover after the completion of claw generation (the first phase). Let $\accept$ denote the event that the verifier accepts at the end of the interaction with $\padv$. We define a ``good'' set $G$ over the support of $W$ as
    \[
    G = \left\{ w : \Pr\left(\accept \mid W=w\right)\geq \frac{3}{4} + \frac{\varepsilon(\sec)}{2} \right\}.
    \]
    We claim that $\Pr\left(W \in G\right) \geq \varepsilon(\sec)/2$, where probability is over the random choices of $\padv$ and the verifier. Assume not, then
    \begin{align*}       
    \Pr\left(\accept\right) &= 
    \Pr\left(\accept\mid W \in G\right) \cdot \Pr\left(W \in G\right) + \Pr\left(\accept\mid W \notin G\right) \cdot \Pr\left(W \notin G\right)\\
    &< \frac{\varepsilon(\sec)}{2} + \frac{3}{4} + \frac{\varepsilon(\sec)}{2} = \frac{3}{4} + \varepsilon(\sec),
    \end{align*}
    in contradiction to the initial assumption on $\padv$. 
    %This means that the internal state of $\padv$ is in the set $G$ with probability at least inverse-polynomial in $n$. Thus we can henceforth fix some $y\in G$ and the state of $\mathcal{A}$ right after sending such $y$.

    Next, we argue that the vector $r$ that the verifier sends in Step~\ref{step:simple-r} of the commitment phase distributes uniformly at random in the eyes of the prover. To see this, observe that $r$ is sampled as follows:
    \begin{itemize}
        \item With probability $1/2$, $r$ is sampled uniformly conditioned on $r^\intercal t = 1$.
        \item With probability $1/2$, $r$ is uniformly sampled from the row-span of $A$. Conditioned on $A$ is rank-$n$, $t$ is the only non-trivial element in $\ker(A)$. Thus, this is equivalent to sampling a uniform $r$, conditioned on $r^\intercal t = 0$.
    \end{itemize}
    Such a distribution is equivalent to uniform since $r^\intercal t=0$ with probability exactly half for a uniformly random $r$.
    
    We describe an adversary $\badv$ that takes as input a $\padv$'s internal state after claw generation, and aims to predict $r^\intercal t$ for a uniformly random $r\gets\FF^{\ell}_2$ (where $\ell=n+1$). $\badv(w)$ performs the following:
    \begin{enumerate}
        \item On input a uniformly sampled $r \gets \FF_2^{\ell}$, $\badv$ sends $r$ to $\padv(w)$ as the message from Step~\ref{step:simple-r} of the commitment phase. 
        \item $\padv$ returns some $d\in \mathbb{F}_2^{\ell}$.
        \item\label{step:B-3} Simulate $\padv$ on both $\theta = \pi/8$ and $\theta = -\pi/8$, rewinding the algorithm in-between. If $\padv$ outputs the same answer for both cases return $0$, else return $1$.
    \end{enumerate}
    
    Let us denote by $\equal(w)$ the event where the answers of simulated $\padv$ in Step~\ref{step:B-3} of $\badv(w)$ are indeed identical for both $\theta = \pi/8$ and $\theta = -\pi/8$. To conclude the proof of the theorem it suffices to show that
    \begin{equation}\label{eq:final}        
    \Pr_{r,\padv}\left(\equal(w)~\middle|~ r^\intercal t =0, w\in G\right) - \Pr_{r,\padv}\left(\equal(w)~\middle|~ r^\intercal t =1, w\in G\right) \geq 2\varepsilon(\sec).
    \end{equation}
    This is indeed sufficient, since it implies the existence of a predictor that has bias $\epsilon$ against $r^\intercal t$ conditioned on $w\in G$. Consequently, by \cref{lmm:GL} (Goldreich-Levin), there exists an extractor with memory $O(\ell\log \ell)$ that outputs $t$, given oracle access to $\badv(w)$ that is successful when $w\in G$. Given such an extractor, we define $\adv$ as follows: Simulate $\padv$ in the claw generation with the verifier and obtain an internal state $w$, then use the extractor with access to $\badv(w)$ to extract $t$. Since $w\in G$ with probability at least $\epsilon/2$, it suffices to argue inverse-polynomial success probability of $\adv$ given $w\in G$, which follows by the success of the Goldreich-Levin extractor.%Further, $\badv$ uses the same memory of $\padv$ up to an additive $O(\ell)$ factor and, therefore, so does $\adv$.
    
    To conclude, we turn to the proof of \cref{eq:final}, which is derived by the following
    \begin{align*}
        \frac{3}{4} +\frac{\varepsilon(n)}{2} &\leq \Pr_{r\in \mathbb{F}_2^{n+1}}\left(\accept\mid w\in G\right) \\
        &= \frac{1}{2}\Pr_{r\in \mathbb{F}_2^{n+1}}\left(\accept\mid r^\intercal t=0, w\in G\right) +
        \frac{1}{2}\Pr_{r\in \mathbb{F}_2^{n+1}}\left(\accept\mid r^\intercal t=1, w\in G\right)\\
        &\leq \frac{1}{4} + \frac{1}{4}\Pr_{r\in \mathbb{F}_2^{n+1}}\left(\equal(w)\mid r^\intercal t =0, w\in G\right) +
        \frac{1}{2}\Pr_{r\in \mathbb{F}_2^{n+1}}\left(\accept\mid r^\intercal t=1, w\in G\right)\\
        &\leq \frac{1}{4} + \frac{1}{4}\Pr_{r\in \mathbb{F}_2^{n+1}}\left(\equal(w)\mid r^\intercal t =0, w\in G\right) +
        \frac{1}{2} - \frac{1}{4}\Pr_{r\in \mathbb{F}_2^{n+1}}\left(\equal(w)\mid r^\intercal t =1, w\in G\right)\\
        &= \frac{3}{4} + \frac{1}{4}\left(\Pr_{r\in \mathbb{F}_2^{n+1}}\left(\equal(w)\mid r^\intercal t =0, w\in G\right) - \Pr_{r\in \mathbb{F}_2^{n+1}}\left(\equal(w)\mid r^\intercal t =1, w\in G\right)\right).
    \end{align*}
\end{proof}

\section{Proof of Quantumness with Arbitrary Gap}\label{sec:exp}

Our second main result is a proof of quantumness protocol that can exhibit up to an exponential gap between the space complexity of the honest parties and the space complexity of the best attack. The protocol instantiates a template laid down by~\cite{BGKPV23} which, similarly to our protocol from \Cref{sec:simple}, follows the general framework from the literature~\cite{BCMVV21,KCVY22,KLVY23} that bases PoQ on claw generation. Unlike the simpler protocol from \Cref{sec:simple}, the claw here is generated in an interactive process, namely via \emph{interactive hashing}~\cite{NOVY}. While the protocol, just like our first one, guarantees only a constant gap between completeness and soundness, the gap can be naturally amplified by sequential repetition. Formally, we obtain the following result.

\begin{theorem}[PoQ with Arbitrary Gap]\label{thm:arbitrary-poq}
    There exists a protocol between a classical verifier and a quantum prover where the two parties take as input a security parameter $\sec\in\NN$ and, at the end of the interaction, the verifier either accepts or rejects, and
    \begin{itemize}
        \item (Completeness) The verifier accepts with probability $\cos^2(\pi/8)$ when interacting with an honest prover.
        \item (Soundness) Any non-uniform classical prover $\padv$ that on input $\sec$ uses less than $m(\sec)$ memory bits has success probability at most $3/4+2^{-\Omega(\sec)}$ in making the verifier accept.
        \item (Complexity) The verifier and (honest) prover run in time $O(\sec^7m^3\cdot\polylog m)$ and in space $O(\sec\cdot\polylog m)$.
    \end{itemize}
\end{theorem}

\subsection{PoQ from Claw Generation}\label{sec:poq-from-claw}

We recall the PoQ outline from~\cite[Figure 4]{BGKPV23}, which assumes the existence of a claw generation sub-protocol. In fact, \cite{BGKPV23} assumes a special case of claw generation, namely claw generation via \emph{trapdoor claw-free functions} (TCF). In contrast, we consider a more general outline where claw generation may be performed arbitrarily as long as it guarantees completeness, i.e. that the prover obtains a claw state, and soundness, i.e. that it is hard for the verifier to compute both values in the claw simultaneously. This does not affect the analysis of the protocol at all: its completeness and soundness follow from those of claw generation, just as in~\cite{BGKPV23}. \iflncs Nevertheless, we provide proofs for the sake of completeness in \Cref{sec:appendix}. \else Nevertheless, we provide proofs for the sake of completeness. \fi 

\begin{protocol}[PoQ from Claw Generation]\label{protocol:template} The following proof of quantumness is parameterized by a security parameter $\sec\in\NN$ and consists of three phases:
\begin{itemize}
    \item (Claw Generation) The prover and verifier engage in a \emph{claw generation sub-protocol} at the end of which the verifier has a pair of \emph{claw values} $x_0,x_1\in\bin^{\ell}$, where $\ell:=\ell(\sec)$ is a polynomial, and the prover's residual state is the \emph{claw state}\begin{equation}\label{eq:claw}
        \frac{1}{\sqrt{2}}(\ket{x_0}+\ket{x_1})_{\X}.
    \end{equation}
    Additionally, we assume that the prover obtains a function $g:\bin^\ell\to \bin$ that satisfies $g(x_0)=0$ and $g(x_1)=1$.
    \item (Commitment)
    \begin{enumerate}
        \item\label{step:r0r1} The verifier samples uniformly random $r_0,r_1\gets\FF_2^{\ell}$ and sends it to the prover.
        \item The prover applies the following isometric mapping to his state
        \[
        \ket{x}_\X \mapsto \ket{x}_{\X}\ket{{r_{g(x)}}^\intercal x}_\B,
        \]
        then measures $\X$ in Hadamard basis to obtain $d \in \mathbb{F}_2^{\ell}$ and sends it to the verifier.
    \end{enumerate}
    \item (CHSH Test) 
    \begin{enumerate}
        \item The verifier samples an angle $\theta \gets \{\pi/ 8, - \pi/8\}$ and sends $\theta$ to the prover.
        \item The prover measures the qubit in $\B$ in the basis
        \[
        \left\{\cos(\theta)\ket{0} + \sin(\theta)\ket{1}, \cos(\theta)\ket{0} - \sin(\theta)\ket{1}\right\}
        \]
        to obtain an outcome $b\in \{0,1\}$ and sends it to the verifier.
        \item \label{step:last-test} Let $\alpha=r_0^\intercal x_0\oplus r_1^\intercal x_1=(r_0||r_1)^\intercal (x_0||x_1)$. The verifier accepts if the following conditions are satisfied (hereby, arithmetics are over the integers):
        \begin{itemize}
            \item If $\theta = \pi/8$, accept iff \begin{equation}\label{eq:b-theta-plus}
                (-1)^b=(1-\alpha)(-1)^{r_0x_0} + \alpha(-1)^{d^\intercal(x_0\oplus x_1)}.
            \end{equation}
            \item If $\theta = -\pi/8$,  accept iff \begin{equation}\label{eq:b-theta-minus}(-1)^b=(1-\alpha)(-1)^{r_0x_0} - \alpha(-1)^{d^\intercal(x_0\oplus x_1)}.
            \end{equation}
        \end{itemize}
        \end{enumerate}
\end{itemize}
\end{protocol}

\iflncs

The following two lemmas, which we prove in \Cref{sec:appendix}, derive the completeness and soundness of the above protocol given the claw generation phase indeed satisfies the desired properties.

\begin{restatable}[Completeness~\cite{BGKPV23}]{lemma}{templatecompleteness}\label{lem:template-completeness}
    The verifier in \Cref{protocol:template} accepts with probability $\cos^2(\pi/8)\approx 0.853$ when interacting with an honest quantum prover.
\end{restatable}

\begin{restatable}[Soundness~\cite{BGKPV23}]{lemma}{templatesoundness}\label{lem:template-soundness}
    Assume there exists a classical adversary $\padv(1^\sec)$ that makes the verifier from \Cref{protocol:template} accept with probability larger than $3/4 + \epsilon(\sec)$ for some function $\epsilon=1/\poly$. Then, there exists an adversary $\adv(1^\sec)$ that interacts with the verifier in the claw generation sub-protocol (the first phase) and, at the end of the interaction, outputs $(x_0,x_1)$ with probability at least $1/\poly(\sec)$, where $x_0,x_1\in\bin^\ell$ are the claw values that the verifier obtains. Further, the space complexity of $\adv$ is larger than the space complexity of $\padv$ by an additive factor of at most $O(\ell\log\ell)$.
\end{restatable}

\else

In the following, we confirm that if the prover in \Cref{protocol:template} completes claw generation with the claw state from \Cref{eq:claw}, then he is able to convince the verifier into accepting with good probability. The proof follows the lines of the proof of Proposition 5.4 in~\cite{BGKPV23}.

\iflncs

\templatecompleteness*

\else

\begin{lemma}[Completeness~\cite{BGKPV23}]\label{lem:template-completeness}
    The verifier in \Cref{protocol:template} accepts with probability $\cos^2(\pi/8)\approx 0.853$ when interacting with an honest quantum prover.
\end{lemma}

\fi

\begin{proof}
    Assuming the completeness of claw generation, the prover's state at the end of the first phase of \Cref{protocol:template} is $(\ket{x_0}+\ket{x_1})/\sqrt{2}$. Let $r_0,r_1\in\mathbb{F}_2^{\ell}$ be the strings sent by the verifier in Step~\ref{step:r0r1} of the commitment phase. The prover maps its state to
    \begin{equation}
        \frac{1}{\sqrt{2}}\left(\ket{x_0}+\ket{x_1}\right)_\X \mapsto \frac{1}{\sqrt{2}}\left(\ket{x_0,  r_0^\intercal x_0}+\ket{x_1, r_1^\intercal x_1}\right)_{\X\B},
    \end{equation}
    then measures $\X$ in the Hadamard basis and gets $d\in \mathbb{F}_2^{\ell}$, which leaves him with
    $$
        \frac{(-1)^{d^\intercal x_0}}{\sqrt{2}}(\ket{r_0^\intercal x_0} + (-1)^{d^\intercal(x_0\oplus x_1)}\ket{r_1^\intercal x_1}).
    $$
    From here, it follows by inspection that, when the above state is measured in the basis $\{\ket{\pi/8},\ket{5\pi/8}\}$ (i.e. when $\Theta=\pi/8$), then $b$ satisfying \Cref{eq:b-theta-plus} is the most likely outcome with probability $\cos^2(\pi/8)$. Otherwise, when the state is measured in the basis $\{\ket{-\pi/8},\ket{3\pi/8}\}$, the most likely outcome is $b$ satisfying \Cref{eq:b-theta-minus}.
    % We distinguish between two cases.
    % \begin{itemize}
    %     \item (Case I) $r^\intercal x_0 = r^\intercal x_1=\delta$ for some $\delta\in\FF_2$: In such a case, the prover's state before measuring $\B$ (RHS of \Cref{eq:claw-state}) is a tensor product between a claw state in $\X$ and the qubit in $\B$, therefore, the measurement of $\X$ has no effect on $\B$. Thus, the prover's residual state after measuring is $\ket{\delta}$. In the CHSH test, the verifier verifies $b$ against the most likely outcome of the prover's measurement, which is $\delta=r^\intercal x_0$. It follows, then, that the prover succeeds with probability exactly $\cos^2(\pi/8)$.
    %     \item (Case II) $r^\intercal x_0 \neq r^\intercal x_1$: Here, the residual state of the prover after measuring $\X$ is
    %     \[
    %         \frac{1}{\sqrt{2}}\left(\ket{0}+(-1)^{d\cdot(x_0 + x_1)}\ket{1}\right).
    %     \]
    %     Thus, if $\theta = \pi/8$, the most likely outcome of the prover's measurement is $d^\intercal (x_0+x_1)$, whereas if $\theta = -\pi/8$, the most likely outcome is $1-d^\intercal (x_0+x_1)$. We conclude that the honest prover succeeds with probability exactly $\cos^2(\pi/8)$ in this case as well.
    % \end{itemize}
\end{proof}

The soundness of \Cref{protocol:template} holds whenever it is hard for any classical prover to compute both values in the claw from \Cref{eq:claw} simultaneously, i.e. $x_0$ and $x_1$.

\iflncs

\templatesoundness*

\else

\begin{lemma}[Soundness~\cite{BGKPV23}]\label{lem:template-soundness}
    Assume there exists a classical adversary $\padv(1^\sec)$ that makes the verifier from \Cref{protocol:template} accept with probability larger than $3/4 + \epsilon(\sec)$ for some function $\epsilon=1/\poly$. Then, there exists an adversary $\adv(1^\sec)$ that interacts with the verifier in the claw generation sub-protocol (the first phase) and, at the end of the interaction, outputs $(x_0,x_1)$ with probability at least $1/\poly(\sec)$, where $x_0,x_1\in\bin^\ell$ are the claw values that the verifier obtains. Further, the space complexity of $\adv$ is larger than the space complexity of $\padv$ by an additive factor of at most $O(\ell\log\ell)$.
\end{lemma}

\fi

\begin{proof}
    The soundness analysis from~\cite{BGKPV23} is in two steps: \begin{enumerate*}[label=(\roman*)]
        \item \label{item:soundness-bgk} First, they observe that a classical prover that succeeds with probability better than $3/4$ can essentially predict the value $\alpha=(r_0||r_1)^\intercal(x_0||x_1)$ used for the verification in Step~\ref{step:last-test} in the CHSH test. \item \label{item:soundness-gl} Second, by a Goldreich-Levin argument similar to that from the proof of soundness of our first protocol (\Cref{lem:simple-soundness}), they show that a predictor of $(r_0||r_1)^\intercal(x_0||x_1)$ for random $r_0,r_1$ can be transformed into an algorithm that computes $x_0,x_1$.
    \end{enumerate*}

    Similarly to the proof of \Cref{lem:simple-soundness}, we begin by fixing an internal state of $\padv$ after the claw generation is complete, with which he succeeds to convince the verifier at the end with good probability. Specifically, letting $W$ denote $\padv$'s state after claw generation (the first phase) and $\accept$ denote the event that the verifier accepts, we again define a ``good'' set $G$ over the support of $W$ as
    \[
    G = \left\{ w : \Pr\left(\accept \mid W=w\right)\geq \frac{3}{4} + \frac{\varepsilon(\sec)}{2} \right\}.
    \]
    By an averaging argument, it holds that $\Pr\left(W \in G\right) \geq \varepsilon(\sec)/2$.

    Next, we proceed with showing~\ref{item:soundness-bgk}. Denote by $b_{+}$ the value of $b$ that satisfies \Cref{eq:b-theta-plus} and by $b_{-}$ the value that satisfies \Cref{eq:b-theta-minus}. Then, it holds that $(-1)^{b_+\oplus b_-} = (1-\alpha)^2-\alpha^2= 1-2\alpha = (-1)^\alpha$. Hence, predicting the parity $b_+\oplus b_-$ is equivalent to predicting $\alpha$. We describe an adversary $\badv$ that takes as input a $\padv$'s internal state after claw generation, and aims to predict $\alpha=(r_0||r_1)^\intercal (x_0||x_1)$ for uniformly random $r_0,r_1\gets\FF^{\ell}_2$. $\badv(w)$ performs the following:
    \begin{enumerate}
        \item On input uniformly sampled $r_0,r_1 \gets \FF_2^{\ell}$, $\badv$ sends $r_0,r_1$ to $\padv(w)$ as the message from Step~\ref{step:r0r1} of the commitment phase.
        \item $\padv$ returns some $d\in \mathbb{F}_2^{\ell}$.
        \item Simulate $\padv$ on both $\theta = \pi/8$ and $\theta = -\pi/8$, rewinding the algorithm in-between. If $\padv$ outputs the same answer for both cases return $0$, else return $1$.
    \end{enumerate}

    We argue that, given $w\in G$, $\badv(w)$ predicts $b_+\oplus b_-$, and therefore $\alpha$, with advantage at least $2\epsilon(\sec)$. This follows by the same derivation with which we proved the advantage of the same predictor in the proof of \Cref{lem:simple-soundness}, namely \Cref{eq:final} (there, the corresponding parity $b_+\oplus b_-$ is precisely the inner product $r^\intercal t$).
    
    The proof is then complete via \ref{item:soundness-gl}, again similarly to the proof of \Cref{lem:simple-soundness}. The predictor $\badv$ against $(r_0||r_1)^\intercal (x_0||x_1)$ implies, via \cref{lmm:GL} (Goldreich-Levin), an algorithm $\adv$ that computes $x_0,x_1$ as follows: $\adv$ simulates $\padv$ in the claw generation with the verifier and obtains an internal state $w$. Then, $\adv$ applies the Goldreich-Levin extractor with access to $\badv(w)$. By \cref{lmm:GL}, $\adv$ uses at most $O(\ell\log \ell)$ additional memory compared to $\padv$.
\end{proof}

\fi

\subsection{Interactive Hashing}

In a (classical) \emph{interactive hashing} protocol~\cite{NOVY}, two parties, say Alice (to be thought of as a challenger, or the verifier later in our context) and Bob (a challengee, or the prover), engage in an interaction the defines a 2-to-1 hash function $h$ over $[k]$ and a hash value $y$. An interactive hashing protocol allows Bob to control one of the pre-images of $y$ under $h$, specifically to make it so $h(v)=y$ for an input $v$ of his choice. However, the soundness of interactive hashing prevents Bob from controlling \emph{both} pre-images.

\begin{definition}[Interactive Hashing~\cite{NOVY}]\label{def:interactive-hashing}
	An \emph{interactive hashing protocol} is a classical protocol between a public-coin Alice (the verifier in our context) and a deterministic Bob (the prover). Alice has no input and Bob has an input $v\in[k]$. For any choice of Alice's public randomness $h$ and Bob's input $v$, we denote by $y=h(v)$ the deterministic answers computed by (an honest) Bob. The protocol satisfies the following:
	\begin{itemize}
		\item (2-to-1 Hash) Any random choice of $h$ is a 2-to-1 function. That is, for any $h$ and $y$, $|h^{-1}(y)|=2$.
		\item ($(\alpha,\beta)$-Security) For any fixed set $B\subseteq[k]$ of size at most $\beta k$, for any (possibly malicious, unbounded) Bob's strategy which, on input public coins $h$ results in an arbitrary $y$ such that $h^{-1}(y)=\{v_0,v_1\}$, it holds that: \[\Pr\left(\{v_0,v_1\}\subseteq B\right)\leq\alpha.\]
	\end{itemize}

	We say that an interactive hashing protocol is \emph{stateless} if Bob is not required to keep an intermediate state between the rounds of the protocol. Namely, if Bob's message at any round is a deterministic function of its input $v$ and Alice's public coins at that round. 
\end{definition}

The seminal work of Naor et al.~\cite{NOVY} devises a simple interactive hashing protocol, where the hash function is a random 2-to-1 linear map, sent by Alice one row at a time. Bob, in turn, answers by multiplying the rows with its input $v$ at every round, hence the protocol is stateless and has a linear number of rounds. Ding et al.~\cite{DHRS04} propose an improvement over the~\cite{NOVY} protocol where Alice sends random hash functions with certain $t$-wise independence properties, and requires only 4 rounds of interaction. The protocol from~\cite{DHRS04} is also stateless as per \Cref{def:interactive-hashing}: at every round, Bob merely applies the functions sent by Alice on his input and sends back the result.

\begin{theorem}[Constant-round Interactive Hashing~\cite{DHRS04}]\label{thm:constant-round-hashing}
	There exists a stateless 4-round interactive hashing protocol that is $(\alpha,\beta)$-secure for any $\beta>0$ with $\alpha=O(\beta\log k)$. The execution of the protocol and the computation of $h^{-1}$ can be done in time and space polylogarithmic in $k$.
\end{theorem}

Morimae and Yamakawa~\cite{MY23} show how to perform the interactive hashing protocol from~\cite{NOVY} in a setting where Bob's input is a quantum state rather than a classical input. They propose a coherent implementation where Bob's state at the end of the protocol is a superposition over all possible classical inputs consistent with the obtained transcript. In the following lemma, we generalize their implementation to any stateless interactive hashing protocol.

\begin{lemma}[Coherent Implementation of Interactive Hashing]\label{lem:coherent-hashing}
	For any stateless interactive hashing protocol over $k$-bit input, there exists a \emph{coherent implementation} of the protocol between a classical Alice and a \emph{quantum} Bob with $k$-qubit input register $V$, satisfying the following:\begin{itemize}
		\item (Correctness) If Bob's state before the protocol is: $$\rho_\mathsf{VW} = \sum_{v\in\bin^k} \ket{v}_\mathsf{V}\ket{\psi_v}_\mathsf{W},$$ and the transcript of the protocol upon its completion is $(h,y)$, then Bob's state at the end of the protocol is:
		$$\tilde{\rho}_{\mathsf{VW}} = \sum_{v:\ h(v)=y} \ket{v}_\mathsf{V}\ket{\psi_v}_\mathsf{W}.$$
		\item (Complexity) The time and space complexity of quantum Bob in the implementation is polynomial in the time and, resp., space complexity of Bob in the classical protocol. 
	\end{itemize} 
\end{lemma}

\begin{proof}
Let $h_i$ and $y_i$ denote Alice's public coins and, respectively, Bob's message at round $i$ of the protocol. Let $y_i\gets f_i(v,h_i)$ denote Bob's computation at round $i$. The coherent interactive hashing over Bob's quantum input performs the protocol as follows. At round $i$, upon receiving Alice's random coins, Bob applies the following mapping to its state
	$$ \ket{v,z}_\mathsf{VY} \mapsto \ket{v,z\oplus f_i(v,h_i)}_\mathsf{VY},$$
where $\mathsf{Y}$ is an additional ancilla register (created new and initialized to $\ket{0}$ at every round). Bob measures the register $\mathsf{Y}$ to obtain $y_i$ and sends it to Alice.

To see why the implementation satisfies correctness, note that we may purify the execution of the protocol as follows: Apply the following for a random choice of $h=(h_1,\dots,h_r)$ $$\ket{v,z_1,\dots,z_r}\mapsto \ket{v,z_1\oplus f_1(v,h_1),\dots,z_r\oplus f_r(v,h_r)}_{\mathsf{VY}_1\dots \mathsf{Y}_r},$$ then measure the registers $\mathsf{Y}_1,\dots,\mathsf{Y}_r$ to obtain $y$ (recall $h$ is sampled independently in Bob's input). In particular, when applied over $\rho_{\mathsf{VW}}$ and $Y_i$ registers that are initiated to ancillas, the above mapping gives
\begin{align*}
	\sum_{v\in\bin^k}\ket{v}_\mathsf{V}\ket{0}_{\mathsf{Y}_1\dots \mathsf{Y}_r}\ket{\psi_v}_\mathsf{W}&\mapsto \sum_{v\in\bin^k}\ket{v}_\mathsf{V}\ket{f_1(v,h_1),\dots,f_r(v,h_r)}_{\mathsf{Y}_1\dots \mathsf{Y}_r}\ket{\psi_v}_\mathsf{W}\\
	&=\sum_{v\in\bin^k}\ket{v}_\mathsf{V}\ket{h(v)}_{\mathsf{Y}}\ket{\psi_v}_\mathsf{W}.
\end{align*}
Evidently, when we measure $\mathsf{Y}$ in the above state and obtain $y$, we obtain the claimed residual state $\tilde{\rho}_{\mathsf{VW}}$.
\end{proof}

%We stress that a coherent implementation for stateful interactive hashing protocols exists as well, albeit one that incurs a blow-up in space proportional to the number of rounds in the protocol (in such an implementation, Bob must remember all public coins received from Alice to uncompute the intermediate state between any two rounds).

\subsection{Claw Generation}

Equipped with the formulation of coherent interactive hashing, we are prepared to present our claw generation protocol. The following lemma immediately implies \Cref{thm:arbitrary-poq} through \Cref{lem:template-completeness,lem:template-soundness}. In addition, we prove that our protocol provides soundness against space-bounded quantum attackers, albeit with inverse-polynomial soundness error. We elaborate on the applications of quantum soundness in \Cref{sec:qhardness}.

\begin{lemma}[Claw Generation]\label{lem:claw-gen}
    Let $m:=m(\sec)$ be a bound on memory. Let $k:=k(\sec)$ be an additional security parameter such that $k/\log k>\sec(m+\Omega(\sec))$. There exists a protocol between a classical verifier and a quantum prover, where both parties take as input a security parameter $1^\sec$ and at the end of the protocol the verifier outputs a pair of values $x_0,x_1\in\bin^{\ell(\sec)}$, that satisfies the following:
    \begin{itemize}
        \item (Correctness) The residual state of the prover at the end of the protocol is
        $$
            \frac{1}{\sqrt{2}} (\ket{x_0}+\ket{x_1}),
        $$
        with probability 1, where $(x_0,x_1)$ is the output of the verifier.
        Additionally, the prover obtains a function $g:\bin^\ell\to \bin$ that satisfies $g(x_0)=0$ and $g(x_1)=1$.
        \item (Claw-Finding Hardness) Any non-uniform classical, or quantum, adversary with memory of at most $m$ bits (resp., qubits) that interacts with the verifier, outputs $(x_0,x_1)$ with probability at most $2^{-\Omega(\sec)}$, respectively $O(\sec^2\sqrt{m/k})$, at the end of the interaction, where $(x_0,x_1)$ is the output of the verifier.
        \item (Complexity) Both the prover and the verifier run in time $O(\sec\cdot k^3)$ in expectation (and with probability all but negligible in $\sec$) and use $O(\sec\cdot\polylog k)$ space.
    \end{itemize}
    In particular, setting $k=\Theta(\sec (m+\sec)\log m)$ gives a protocol with classical soundness (with negligible soundness error), runtime $O(\sec m^3\polylog m)$ and space $O(\sec\polylog m)$. Setting $k=\Theta(\sec^{2c} m)$ gives a protocol with quantum soundness up to soundness error $1/\sec^c$, polynomial runtime and space $\polylog m$.
\end{lemma}

The rest of this subsection is dedicated to presenting the protocol underlying \Cref{lem:claw-gen} and analyzing its correctness and complexity. The proof of security (hardness of finding a claw) is deferred to \Cref{sec:chardness,sec:qhardness}.

\paragraph{The Protocol.} Our claw generation protocol is parametrized by a \emph{stream length} $k:=k(\sec)$ which is chosen such that $k/(\sec\log k) > m+\Omega(\sec)$. We assume for notational convenience that $k$ is always a power of $2$. The protocol makes use of a coherent interactive hashing sub-routine (see \cref{def:interactive-hashing,lem:coherent-hashing}), which in particular may be based on the protocol from \cref{thm:constant-round-hashing}.

The protocol proceeds as follows:

\begin{enumerate}
    \item (1-bit Claws Generation) For $j\in[\sec]$, the prover creates new $\log k$-qubit register $\V_j$ and 1-qubit register $\U_j$ and engages in the following interaction with the verifier:
    \begin{enumerate}[label*=\arabic*.]
        \item (Setup) The verifier samples two uniform indices $v_0, v_1 \gets [k]$ and the prover prepares the uniform superposition over $\V_j$ and the $0$ state in $\U_j$:
        \[
            \ket{\psi} = \sum_{v\in[k]} \ket{v}_{\V_j} \otimes \ket{0}_{\U_j}.
        \]
        \item\label{step:streaming} (Streaming) The verifier uniformly samples and transmits to the prover a stream of $k$ bits $U=(U_1,\dots,U_k)$, one bit at a time. The verifier stores $(v_0,U_{v_0})$ and $(v_1,U_{v_1})$ in his memory.
        Upon receiving the $i$-th bit $U_i$, the prover applies to its state the unitary defined by the following map:
        \[
        \ket{v}_{\V_j} \otimes \ket{u}_{\U_j} \mapsto \ket{v}_{\V_j}\ket{u \oplus f_{i,U_i}(v)}_{\U_j} \quad\text{where}\quad f_{i,U_i}(v) = \begin{cases}
            U_i & \text{if } v = i \\
            0 & \text{otherwise.}
        \end{cases}
        \]
        \item \label{step:interactive-hashing} (Interactive Hashing) The verifier and the prover engage in a run of the coherent interactive hashing protocol over $[k]$, where the input of the prover is the register $\V_j$. Let $(h,y)$ be the transcript of the interactive hashing upon its completion. If $h^{-1}(y)\neq\{v_0,v_1\}$, the verifier aborts and starts over with the Setup phase. Otherwise, the verifier stores:
        \begin{align*}
            v_0^j=v_0&& v_1^j=v_1&& z_0^j=U_{v_0}&& z_1^j=U_{v_1}.
        \end{align*}
    \end{enumerate}
    \item (Amplification) The prover and the verifier engage in the following protocol:
    \begin{enumerate}[label*=\arabic*.]
    	\item For any $j\in[\sec]$, the verifier arbitrarily picks a function $g^j:[k]\to\bin$ with short description satisfying $g^j(v^j_b)=b$ for $b\in\bin$,\footnote{For instance, $g^j$ can be a dictator function with description size $\log\log k + 1$.} The verifier sends $(g^1,\dots,g^\sec)$ to the prover.
    	\iffalse
    	\begin{comment}
    	\item For any $j\in[\sec]$, the prover adds a 1-qubit ancilla $\B_j$ and applies the following mapping to the registers $\V_j\U_j\B_j$:
    	$$
    	\ket{v}_{\V_j}\otimes \ket{x}_{\U_j}\otimes\ket{b}_{\B_j}\mapsto \ket{v}_{\V_j}\otimes \ket{x}_{\U_j}\otimes\ket{b\oplus g^j(v)}_{\B_j}.
    	$$   	
    	\item For $j=2,\dots,\sec$, the prover adds a new 1-qubit ancilla register $\C_j$ and applies the following mapping over $\B_1,\B_j,\C_j$:
   		$$
   		\ket{b^1}_{\B_1}\otimes \ket{b^j}_{\B_j}\otimes\ket{c}_{\C_j}\mapsto \ket{b^1}_{\B_1}\otimes \ket{b^j}_{\B_j}\otimes\ket{c\oplus b^1\oplus b^j}_{\C_j}.
   		$$
   		\end{comment}
   		\fi
   		\item \label{step:stitch-1} For $j=2\dots\sec$, the prover adds a new 1-qubit ancilla register $\B_j$ and applies the following mapping over $\V_1\V_j\B_j$:
   		$$
   		\ket{v^1}_{\V_1}\otimes \ket{v^j}_{\V_j}\otimes\ket{b}_{\B_j}\mapsto \ket{v^1}_{\V_1}\otimes \ket{v^j}_{\V_j}\otimes\ket{b\oplus g^1(v^1)\oplus g^j(v^j)}_{\B_j}.
   		$$
        \item \label{step:stitch-2} The prover measures $\B_2\dots\B_\sec$ to obtain an outcome $(b_2,\dots,b_\sec)\in\bin^{\sec-1}$, which he sends to the verifier. The verifier outputs:
        \begin{align*}
        	x_0=(v^1_0,z^1_0,v^2_{b_2},z^2_{b_2},\dots,v^\sec_{b_\sec},z^\sec_{b_\sec}) && x_1=(v^1_1,z^1_1,v^2_{1-b_2},z^2_{1-b_2},\dots,v^\sec_{1-b_\sec},z^\sec_{1-b_\sec}).
        \end{align*}
    \end{enumerate}
\end{enumerate}

\paragraph{Complexity.} The runtime of one attempt of a 1-bit claw generation is linear in $k$ for both parties (note interactive hashing is performed over the domain $[k]$ which is of polylogarithmic size). An attempt succeeds with probability $1/k^2$ since the verifier samples $v_0,v_1$ uniformly and independently of the other random choices in the protocol. Therefore, in expectation, the claw generation for a single bit takes $O(k^2)$ time. It also takes $O(k^2)$ time with probability all but negligible in $k$ via a standard Chernoff argument (\Cref{lmm:chernoff}). Since 1-bit claw generation is repeated $\sec$ times, this results in $O(\sec k^3)$ runtime in total, which easily dominates over the runtime of the amplification phase. As for space complexity, observe that, for each of the $\sec$ bits in the claw, both parties require memory at most polylogarithmic in $k$ since they both store a constant amount of variables over $[k]$ and apply interactive hashing over this domain.

\paragraph{Correctness.} 
We show that the output of the protocol results into a well-formed claw (see \Cref{lem:claw-gen}), with certainty, provided that the protocol terminates.

Consider a prover that follows the protocol honestly. At any round $j\in[\sec]$, once the streaming in Step~\ref{step:streaming} is complete, the prover obtains the superposition $\sum_{v\in[k]}\ket{v}\ket{U_v}$ in the registers $\V_j\U_j$, where $U$ is the communicated stream at that round. By \cref{lem:coherent-hashing}, the following coherent interactive hashing step, if succeeds, results in the state $\ket{v^j_0}\ket{z^j_0}+\ket{v^j_1}\ket{z^j_{1}}$ in these registers. Notice that the generation of these 1-bit claws is independent across all $j$ and, hence, the prover's state prior to the stitching stage may be written as:
\[ \left(\sum_{\beta\in\bin}\ket{v^1_\beta}_{\V_1}\ket{z^1_\beta}_{\U_1}\right)\otimes\dots\otimes\left(\sum_{\beta\in\bin}\ket{v^\sec_\beta}_{\V_\sec}\ket{z^\sec_\beta}_{\U_\sec}\right).\]
By inspection, since $g^1(v^1_{\beta_1})\oplus g^j(v^j_{\beta_j})=\beta_1\oplus \beta_j$ for any $j$, the mapping performed at Step~\ref{step:stitch-1} gives:
\[ \sum_{\beta_1\in\bin}\ket{v^1_{\beta_1}}_{\V_1}\ket{z^1_{\beta_1}}_{\U_1}\otimes
\bigotimes_{j=2\dots \sec}
\left(\sum_{\beta_j\in\bin}\ket{v^j_{\beta_j}}_{\V_j}\ket{z^j_{\beta_j}}_{\U_j}\ket{\beta_1\oplus\beta_j}_{\B_j}\right).\]
Thus, when applying the measurement at Step~\ref{step:stitch-2} and obtaining $(b_2,\dots,b_\sec)$, the prover is left with the following state:
$$ \sum_{\beta_1\in\bin}\ket{v^1_{\beta_1}}_{\V_1}\ket{z^1_{\beta_1}}_{\U_1}\otimes \ket{v^2_{\beta_1\oplus b_2}}_{\V_2}\ket{z^2_{\beta_1\oplus b_2}}_{\U_2} \otimes\dots\otimes\ket{v^\sec_{\beta_1\oplus b_\sec}}_{\V_\sec}\ket{z^\sec_{\beta_1\oplus b_\sec}}_{\U_\sec},$$
which is a superposition of the two values $x_0$ and $x_1$ that the verifier outputs.

Lastly, recall that we require the prover to obtain a function $g:\bin^\ell\to\bin$ that differentiates between the two values of the claw by satisfying $g(x_b)=b$. This function may be set to be $g^1$ (which the prover receives from the verifier), applied to the first $\log k$ bits of $x_b$, namely the part containing the index $v^1_b$.

\subsection{Classical Hardness of Finding a Claw} \label{sec:chardness}

We show a bound on the success probability of any \emph{classical} prover to output a claw. For $j\in[\sec]$, let $U^j$ be the random variable taking the value of the string $U$ streamed in Step~\ref{step:streaming}, before the first successful completion of the interactive hashing step, namely the first attempt where the verifier does not abort and start over in Step~\ref{step:interactive-hashing}. Let $W^j_{pre}$ denote the random variable consisting of the adversary's internal state, i.e., its $m$-bit memory, right after the streaming of $U^j$ has completed and just before the start of the interactive hashing. Let $W^j_{post}$ denote its internal state after the completion of the interactive hashing. For $b\in\bin$, we let the random variable $Z^j_b$ take the value of the bit $z^j_b$ that is stored by the verifier in Step~\ref{step:interactive-hashing} at round $j$ and let $Z_b=(Z^1_b,\dots,Z^\sec_b)$ (note these values are part of the claw $x_0,x_1$). %Additionally, for any $j$, let $Z^{<j}_b$ denote the $(j-1)$-bit prefix of ${Z}_b$.

By definition, the following lemma implies the classical hardness of finding a claw.

\begin{lemma}[Classical Hardness of Claw-Finding]\label{lem:claw-gen-classical-soundness}
    It holds that $H_\infty(Z_0,Z_1\mid W^\sec_{post})\geq \Omega(\sec).$
\end{lemma}

Let $t=k/(\sec\log k)-m=\Omega(\sec)$. Denote by $\Bad_j$ the event that $W^j_{pre}=\omega$ for $\omega\in\bin^m$ satisfying:% $Z^{<j}_0=\zeta_0$ and $Z^{<j}_1=\zeta_1$ for $\omega\in\bin^m$ and $\zeta_0,\zeta_1\in\bin^{j-1}$ satisfying:
\begin{equation}\label{eq:bad}
    H_\infty(U^j\mid W^j_{pre}=\omega, W^j_{post})<k-(m+t).\footnote{Note this is the conditional entropy of the marginal distribution of $U_j$ when $W^j_{pre}=\omega$, conditioned on the random variable $W^j_{post}$.}
\end{equation}
Let us first bound the probability that the above bad event occurs for any round $j$.
\begin{claim}[Bad Events]\label{prop:bad-bound}
	$\Pr\left(\bigcup_j\Bad_j\right)<2^{-\Omega(\sec)}$.
\end{claim}
\begin{proof}
	By \Cref{prop:independent-condition}, since $U^j$ is independent in $W^j_{post}$ given $W^j_{pre}$, it holds that $H_\infty(U^j\mid W^j_{pre}=\omega, W^j_{post})\geq H_\infty(U^j\mid W^j_{pre}=\omega)$. Thus, for any $j$, we have
	$$\Pr\left(\Bad_j\right) < 2^{k-(m+t) - H_\infty(U^j\mid W^j_{pre})}<  2^{k-t - H_\infty(U^j)}=2^{-t},$$
	where the first inequality is by \cref{lem:min-entropy-tail-bound} and the second by \cref{lem:min-entropy-conditional}. A union bound is then sufficient to derive the inequality.
\end{proof}
In the next proposition, we invoke the soundness of the interactive hashing to argue the unpredictability of \emph{both} $Z^j_0$ and $Z^j_1$, for any round $j$, even given the adversary's bounded storage.

\begin{claim}[Entropy Bound]\label{prop:high-bit-entropy}
    Assume $\Bad_j$ does not occur, and fix any $W^j_{pre}=\omega$ in the support.%,Z_0^{<j}=\zeta_0,Z_1^{<j}=\zeta_1$ in the support.
    Then, it holds that:
    \[
    \Pr\left(H_\infty(Z^j_0,Z^j_1\mid \omega,W^j_{post})> 0.8\right) =1-O(1/\sec) 
    \]
    where the probability is taken over the randomness of the interactive hashing at round $j$.
\end{claim}

\begin{proof}
Conditioned on $\Bad_j$ not occurring, fix any $W^j=\omega$ %, $Z^{<j}_0=\zeta_0$ and $Z^{<j}_1=\zeta_1$
in the support. We have:
\begin{align*}
    \Expct_{i\gets[k]}\left(H(U^j_i\mid W^j_{pre}=\omega,W^j_{post})\right)&=\frac{1}{k}\sum_{i\in[k]} H(U^j_i\mid \omega,W^j_{post})\\
    & \geq \frac{1}{k} H_\infty(U^j \mid \omega,W^j_{post})\\
    &\geq \frac{1}{k}(k-(m+t))=1-(m+t)/k,
\end{align*}
where the first inequality follows by \cref{prop:min-entropy-block-entropy} and the second one follows be the definition of the $\Bad_j$ event in \cref{eq:bad}. 
Define the set:
$$
B^j_{\omega}:=\left\{i\in[k]\mid H(U^j_i\mid W^j_{pre}=\omega,W^j_{post})< 0.99\right\}.
$$
By an averaging argument, it must hold that $\Pr_{i\gets[k]}(i\in B^j_{\omega})\leq 100(m+t)/ k$, and consequently that $|B^j_{\omega}|\leq 100(m+t)$. Thus, by \cref{thm:constant-round-hashing}, we have that:
\[\Pr\left(\{v^j_0,v^j_1\}\subseteq B^j_{\omega}\right)= O((m+t)\log k/ k)=O(1/\sec),\] where probability is over the randomness of the interactive hashing in the choice of $\{v^j_0,v^j_1\}$ and it holds for any (possibly unbounded) prover. By definition, we have:
\[
\Pr\left(\exists b\in\bin:\ H(Z^j_b\mid W^j_{pre}=\omega,W^j_{post})=0.99\right)>1-O(1/\sec).
\]
From \cref{prop:entropy-to-min-entropy}, this implies:
\[\Pr\left(\exists b\in\bin:\ H_\infty(Z^j_b\mid W^j_{pre}=\omega,W^j_{post})>0.8\right)=1-O(1/\sec),\]
and, in particular, the proposition follows.
%\textcolor{red}{GM: I think the $>$ above should be $\geq$, but I don't think it matters.}
\end{proof}
%\textcolor{red}{GM: Maybe here we can have a formal lemma statement and wrap the thing below in a proof environment.}

Fix any sequence $\omega=(\omega^1,\dots,\omega^\sec)$ of values taken by $W^1_{pre},\dots,W^\sec_{pre}$ throughout the protocol. Define the random variable:
\[
J_{\omega} := \left\{j : H_\infty(Z^j_0,Z^j_1\mid W^j_{pre}=\omega^j,W^j_{post})<0.8 \right\} \subseteq[\sec].
\]

\begin{claim}\label{prop:J-bound}
    Assume $\Bad_j$ does not occur for any $j$. Then, it holds that $\Pr(|J_\omega|\geq\sec/2)<2^{-\Omega(\sec^2)}$ for any $\omega$ in the support, where probability is over the randomness of the interactive hashing invocations.
\end{claim}

\begin{proof}
Assuming that $\omega$ are such that $\Bad_j$ does not hold for any $j$, we have that
\[
\Expct \left(|J_{\omega}|\right) = \sec \cdot O(1/\sec)=O(1)
\]
by \cref{prop:high-bit-entropy} (expectation is taken over the interactive hash invocations).
%\tamer{make sure this is the case} \textcolor{red}{GM: what is the expectation taken over?}
Further, notice that upon fixing $\omega$, the membership of the different $j$'s in $J_{\omega}$ constitute anti-correlated predicates, since the probability of each is bounded by the soundness of the interactive hash, independently of what happens in other rounds. Hence, we may invoke Chernoffs's inequality (\Cref{lmm:chernoff}) to obtain the following tail bound on the size of $J_{\omega}$:
\[
    \Pr\left(|J_{\omega}|\geq\sec/2\right)< e^{-\Omega(\sec^2)}=e^{-\Omega(\sec^2)}.
\]
\end{proof}

%\textcolor{red}{GM: Maybe we can recall Hoeffding bound. I changed $x_b $ with $Z_b$ and $\chi$ with $\zeta$, up to here, and I am on board with everything written until here. I have not checked the equations below.}\tamer{can you check now?}

We consider a run of the protocol conditioned on $\Bad_j$ not occurring for any $j$ and $|J_\omega|< \sec/2$, where $\omega=(\omega^1,\dots,\omega^\sec)$ are the values taken by $W^1_{pre},\dots,W^\sec_{pre}$. By \Cref{prop:bad-bound,prop:J-bound}, if an adversary succeeds in the original experiment with certain probability, he must succeed in the conditional experiment with probability smaller by at most $2^{-\Omega(\sec)}$. It remains, then, to bound the success probability under the above two conditions.

We denote the distributions in the conditional experiment by $\tilde{Z}$, $\tilde{W}_{pre}$ and $\tilde{W}_{post}$. Note that $\tilde{Z}^j_0,\tilde{Z}^j_1$ are independent in all other values in $\tilde{Z}_0$ and $\tilde{Z}_1$ given the snapshots of adversary's state before and after they are determined, namely $\tilde{W}^j_{pre}$ and $\tilde{W}^j_{post}$. Hence, by \Cref{prop:independent-condition}, 
\begin{align*}
	H_\infty(\tilde{Z}_0,\tilde{Z}_1\mid \tilde{W}^\sec_{post}) &\geq H_\infty(\tilde{Z}_0,\tilde{Z}_1\mid \tilde{W}^\sec_{post},\dots,\tilde{W}^1_{post},\tilde{W}^\sec_{pre},\dots,\tilde{W}^1_{pre})\\
    & \geq \sum_{j=1}^\sec H_\infty(\tilde{Z}^j_0,\tilde{Z}^j_1\mid \tilde{W}^\sec_{post},\dots,\tilde{W}^1_{post},\tilde{W}^\sec_{pre},\dots,\tilde{W}^1_{pre})\\
    & \geq \sum_{j=1}^\sec H_\infty(\tilde{Z}^j_0,\tilde{Z}^j_1\mid \tilde{W}^j_{post},\tilde{W}^j_{pre})\\
    & \geq \min_\omega \sum_{j\in J_\omega} H_\infty(\tilde{Z}^j_0,\tilde{Z}^j_1\mid \tilde{W}^j_{post},\tilde{W}^j_{pre}=\omega^j)\\
    & \geq 0.8(\sec/2).
\end{align*}

The above implies that the probability of success in the hybrid experiment is at most $2^{-0.4\sec}=2^{-\Omega(\sec)}$ and, consequently, is at most $2^{-\Omega(\sec)}$ in the original experiment as well. This completes the proof of \Cref{lem:claw-gen-classical-soundness}.

\iffalse
\begin{align*}
    H_\infty(Z_0,Z_1\mid W^\sec_{post}) &\geq -\sum_j \log \Expct_{\omega\gets W^\sec_{pre},\omega'\gets W^\sec_{post}}[\max_{{\zeta}_0,{\zeta}_1} \Pr({Z}_0={\zeta}_0,{Z}_1={\zeta}_1\mid \omega,\omega') ]\\
	&= -\log \max_{{\zeta}_0,{\zeta}_1,\omega} \Pr({Z}_0={\zeta}_0,{Z}_1={\zeta}_1\mid W^\sec=\omega)\\
	&= - \max_{{\zeta}_0,{\zeta}_1,\omega} \log  \prod_{j=1}^\sec \Pr({Z}^j_0={\zeta}^j_0,{Z}^j_1={\zeta}^j_1\mid W^\sec=\omega,Z^{<j}_0=\zeta^{<j}_0,Z^{<j}_1=\zeta^{<j}_1)\\
    &= - \max_{{\zeta}_0,{\zeta}_1,\omega} \sum_{j=1}^\sec \log \Pr({Z}^j_0={\zeta}^j_0,{Z}^j_1={\zeta}^j_1\mid W^\sec=\omega,Z^{<j}_0=\zeta^{<j}_0,Z^{<j}_1=\zeta^{<j}_1)\\
    &\geq -(\sec/2)\cdot \max_{{\zeta}_0,{\zeta}_1,\omega,j\in J_{\omega,\zeta_0,\zeta_1}} \log \Pr({Z}^j_0={\zeta}^j_0,{Z}^j_1={\zeta}^j_1\mid \omega,\zeta^{<j}_0,\zeta^{<j}_1).
\end{align*}
\fi

\subsection{Quantum Hardness of Finding a Claw} \label{sec:qhardness}

We prove hardness of finding a claw also against a quantum (but memory-bounded) attacker, with a somewhat worse bound. Although this statement is not necessary for the proof of our proof of quantumness theorem (\Cref{thm:arbitrary-poq}), it enables new applications in the context of verification of quantum computation, which we outline in \cref{sec:appendix}. Before starting with the analysis, let us make the notion of a memory-bounded quantum adversary more precise.

\paragraph{Quantum Adversaries.} A memory-bounded quantum adversary is modeled as a quantum channel acting on a fixed-size register $\mathsf{M} \simeq \mathbb{C}^{2^m}$ and on a register $\mathsf{N}$, which corresponds to the next message of the protocol. We can model any memory-bounded quantum adversary without loss of generality as follows: 
\begin{itemize}
    \item The adversary starts with an initial state $\rho_\mathsf{M}$ in the memory register.
    \item For each round $i$ of the protocol and each incoming message $\mu_i$, the adversary applies an arbitrary CPTP linear map:
\[
\Phi^{i,\mu_i}_{\mathsf{M}\to \mathsf{MN}} : \text{L}(\mathsf{M}) \mapsto\text{L}(\mathsf{M}\otimes \mathsf{N}).
\]
\item The message sent by the adversary as a response is determined by measuring $\mathsf{N}$ in the computational basis.
\item The updated state of the attacker is the reduced density matrix on $\mathsf{M}$.
\end{itemize}
An implication of this fact is that the initial state of the adversary, along with transcript of the protocol, uniquely determine the state of the attacker at any given round.

\paragraph{Analysis.}
Just like in the proof for classical soundness in \Cref{sec:chardness}, let us denote by $Z$ the random variable containing the $2\sec$-bit bits $\{z^1_b,\dots,z^{\sec}_b\}_{b\in \{0,1\}}$ that are stored by the verifier in Step~\ref{step:interactive-hashing}, and by $W$ the random variable containing the $m$-qubit state of the attacker at the end of the protocol. Further, for a set $J\subseteq [\sec]\times\bin$, we denote by $Z_J$ the restriction of $Z$ to $J$.

\begin{lemma}[Quantum Hardness of Claw-Finding]\label{lmm:qhard}
% Except with probability $\exp(-\Omega(\sec))$, there exists a set $J\subset [\sec]\times\bin$ of size $|J| \geq \sec/2$ (which in particular depends on $Z$ and $W$) such that
% \[
% \TD\left((J,Z_J,W),(J, Z',W)\right) \leq \sec^2\sqrt{m/2k},
% \]
% where $Z'$ is a uniformly random string of length $|J|$.
There exists a random variable $J\subseteq [\sec]\times\bin$ (which depends on $Z$ and $W$) such that
\[
\TD\left((J,Z_J,W),(J, Z',W)\right) \leq \sec^2\sqrt{m/2k},
\]
where $Z'$ is a uniformly random string of length $|J|$, and, further, $|J| \geq \sec/2$ with probability all but $\exp(-\Omega(\sec))$ over the random coins of the protocol.
\end{lemma}

Before proving \cref{lmm:qhard}, we first observe that it indeed implies quantum hardness of finding a claw. This is because an adversary that guesses the claw must in particular guess any subset of its bits, including $J$. However, the success probability of the best attacker on the RHS distribution is at most $2^{-\sec/2}$. By a triangle inequality, this implies that the success probability in guessing $Z$ is bounded by $2^{-\sec/2} + \sec^2\sqrt{m/2k}$.%\tamer{Double check the parameters}
% \textcolor{red}{GM: Note that $t$ grows with $\sec$. Should be fine, but double check.}

We now proceed to prove \Cref{lmm:qhard}. Similarly as above, for $j\in[\sec]$, let $U^j$ be the random variable taking the value of the string $U$ streamed in Step~\ref{step:streaming}, before the first successful completion of the interactive hashing step, namely the first attempt where the verifier does not abort and start over in Step~\ref{step:interactive-hashing}. Let $W^j_{pre}$ denote the random variable consisting of the adversary's internal state, i.e., its $m$-qubit memory, right after the streaming of $U^j$ has completed. We consider the marginal distributions of $U^j$ and $W^j_{pre}$ given a fixed transcript of the protocol up to (and excluding) the streaming of $U^j$ (recall the transcript, together with the adversary's initial state, determine the state of the adversary and, therefore the aforementioned marginal distributions are well-defined). For any such possible transcript $\tau$, we denote the corresponding marginals by $U^j(\tau)$ and $W^j_{pre}(\tau)$. By \cref{lem:plug-in}, for any $\tau$, we have that:
\begin{equation}\label{eq:plugin}
    \Expct_{i\gets[k]} \underbrace{\TD\left( (U_{<i}^j(\tau), U_i^j(\tau), W^j_{pre}(\tau)), 	(U_{<i}^j(\tau), \tilde{U}_i^j, W^j_{pre}(\tau)) \right)}_{\delta_i^j(\tau)} \leq \sqrt{m/2k}
\end{equation}
where $\Tilde{U}_{i}^j$ is a uniformly sampled bit. Let $t=\sqrt{2k/\sec^2m}$. For all $j\in[\sec]$ and any $\tau$ define the set:
\[
    B^j_\tau := \left\{ i: \delta_i^j(\tau) > 1/t \right\}\subseteq[k].
\]
Using this definition, we define the set $J$ as follows: For any round $j\in[\sec]$, we let $\tau^j$ denote the history of the protocol up to (and excluding) the streaming of $U^j$. We add $(j,b)$ to $J$ if $b\in\bin$ is the smallest such that $v^j_b\notin B^j_{\tau^j}$, where $v^j_0,v^j_1$ are the outcome of the (successful) interactive hashing at round $j$. To complete the proof of \Cref{lmm:qhard}, it suffices the bound the size of $J$, and the trace distance between the following two experiments:
\iflncs
    The first experiment runs the protocol and outputs the adversary's state $W$ at its completion and the bits in $Z_J$. The second experiment does the same thing, except that it outputs uniformly sampled bits, along with $W$.
\else
\begin{itemize}
    \item The first experiment runs the protocol and outputs the adversary's state $W$ at its completion and the bits in $Z_J$.
    \item The second experiment does the same thing, except that it outputs uniformly sampled bits, along with $W$.
\end{itemize}
\fi
We do so in the following two propositions.
% To prove the bound on trace distance from the statement of \Cref{lmm:qhard}, we bound the distance between the following two experiments. The first experiment runs the protocol and outputs the adversary's state $W$ at its completion and a certain subset of the bits in $Z$ that are remembered and output by the verifier. In the second experiment, we do the same thing except instead of outputting bits from $Z$, we output uniformly random bits. We choose which bits from $Z$ to output, which we define by a subset $J\subseteq[\sec]\times\bin$, as follows. For any round $j\in[\sec]$, we let $\tav^j$ denote the state of the protocol just before the streaming of $U^j$ begins. We add $(j,b)$ to $J$ if $b\in\bin$ is the smallest such that $v^j_b\notin B^j_{\tau^j}$, where $v^j_0,v^j_1$ are the outcome of the (successful) interactive hashing at round $j$.

% To complete the proof of \Cref{lmm:qhard}, we need to show two things: That $J$ is large enough with sufficient probability and that swapping the bits of $Z_J$ with random gives a distribution that is sufficiently close the original distribution. We do so in the following two propositions.

\begin{claim}[$J$ is Large] Let $J$ be defined as above, then: 
\[\Pr(|J|< \sec/2)<2^{-\Omega(\sec)}.\]
\end{claim}

\begin{proof}
    By \Cref{eq:plugin} and \cref{lmm:markov} (Markov) we have that: \[
        \Pr_{i\gets[k]} \left(\delta_{i}^j(\tau^j) > 1/t \right) < t/\sqrt{m/2k},
    \] which implies that $|B^j_{\tau^j}| < t \sqrt{{km/2}}$. By \cref{thm:constant-round-hashing}, we can bound the probability that the pre-images of the interactive hashing belong to such a set by
    \begin{equation}\label{eq:boundround}        
    \Pr\left(\{v_0^j, v_1^j\} \subseteq B^j_{\tau^j}\right)\leq O\left( t \log k\sqrt{{m/2k}} \right),
    \end{equation}
    where the probability is taken over the random coins of the interactive hashing.
    
    Let us denote by $E^j$ the predicate for $\{v_0^j, v_1^j\} \subseteq B^j_\tau$. To prove the proposition, it suffices to show that $E^j=0$ for at least half of the $j$'s with overwhelming probability. %Note that this event is well-defined given the state of the attacker before the streaming at round $j$, since the set $B^j_{\tau^j}$ is also well-defined.
    By \cref{eq:boundround}, we have:
    \[
    \tilde{E}:=\Expct\left(\sum_{j\in[\sec]} E^j\right) =O\left( \sec t \log k \sqrt{{m/2k}}\right) = O(\log k).
    \]
    Further, observe that the random variables $E^j$ are negatively correlated, since interactive hashing soundness holds for any round independently of the others and, therefore, the adversary cannot increase probability that the event $E^j$ happens across several rounds. Thus, by \cref{lmm:chernoff} (Chernoff) we can bound:
    \[
    \Pr\left(\sum_{j\in[\sec]} E^j_\tau \geq \sec/2\right) \leq e^{-\Omega(\sec^2/\log\sec)}
    %\Pr\left(\sum_{j\in[\sec]} E^j_\tau \geq \sec/2\right) \leq e^{-\Omega{\tilde{E}_\tau v^2}{3}}
    \]
    which completes the proof of the claim.
\end{proof}

\begin{claim}[Distance of the Experiments] Let $J$ be defined as above. It holds that:
\[
\TD((J,Z_J,W),(J,Z',W))\leq\sec/t
\]
where $Z'$ is a uniformly random string of length $|J|$.
\end{claim}
\begin{proof}
    Let us denote $J=\{(j,b_j)\}$ (note for every $j$ there exists at most one element in $J$ for some $b_j\in\bin$). We bound the distance incurred by each of the swaps. Let $J^{\geq j} = \{(j',b_{j'})\in J\mid j'\geq j\}$ and let $t_j=|J|-|J^{\geq j}|$.  Then, for any $j=1,\dots,\sec$, our goal is to prove
    $$
        \TD((J,Z'_{t_j},Z_{J^{\geq j}},W),(J,Z'_{t_j+1},Z_{J^{\geq j+1}},W))\leq 1/t,
    $$
    which implies the claim by triangle inequality. Now, for any $(j,b_j)\in J$, let $v^j=v^j_{b_j}$; this is the location of the swapped bit in $U^j$. Then, by monotonicity of trace distance, it holds that
    \begin{align*}
        &\TD((J,Z'_{t_j},Z_{J^{\geq j}},W),(J,Z'_{t_j+1},Z_{J^{\geq j+1}},W))\\
        &\indent \leq \TD((J^{\leq j},Z'_{t_j},U^j_{<v^j},U^j_{v^j},W^j_{pre}),(J^{\leq j},Z'_{t_j},U^j_{<v^j},\tilde{U}^j_{v^j},W^j_{pre})).
    \end{align*}
    The above holds since the prover's final state $W$, the final set $J$ and the bits $Z_J$ swapped with random until round $j-1$ (i.e. $Z'_{t_j},Z_{J^{\geq j}}$) or, respectively, round $j$ (i.e. $Z'_{t_j+1},Z_{J^{\geq j+1}}$), may be produced from the prover's state after the streaming of $U^j$, i.e. $W^j_{pre}$, the choice of $J$ up until the $j^{th}$ round, i.e. $J^{\leq j}$, $Z_J$ swapped up to round $j-1$, i.e. $Z'_{t_j}$, and the stream $U^j$ up until the $v_j^{th}$ bit which is either swapped (i.e. $\tilde{U}^j_{v^j}$) or not ($U^j_{v^j}$), respectively. To produce the final distributions, simply carry on with the protocol given $W_{pre}^j$ and $J^{\leq j}$ to obtain the final state $W$ and $J$ in full, together with the bits of $Z_J$ at rounds $j+1,\dots,\sec$. The bit in $Z_{J}$ at round $j$ is simulated by $U^j_{v^j}$ or $\tilde{U}^j_{v^j}$ and the bits before round $j$ are given as $Z'_{t_j}$.

    Lastly, since anything that occurs up till round $j$ is a function of the transcript of the protocol up to that round, i.e. $\tau^j$, we may finish as follows
    \begin{align*}
        &\TD((J^{\leq j},Z'_{t_j},U^j_{<v^j},U^j_{v^j},W^j_{pre}),(J^{\leq j},Z'_{t_j},U^j_{<v^j},\tilde{U}^j_{v^j},W^j_{pre}))\\
        &\indent \leq \Expct_{\tau^j} \TD((U^j_{<v^j}(\tau^j),U^j_{v^j}(\tau^j),W^j_{pre}(\tau^j)),(U^j_{<v^j}(\tau^j),\tilde{U}^j_{v^j},W^j_{pre}(\tau^j)))\\
        &\indent= \Expct_{\tau^j} \delta^j_{v^j}(\tau^j)\leq 1/t.
    \end{align*}
\end{proof}

\iflncs
\bibliographystyle{splncs04}
\else
\bibliographystyle{alpha}
\fi
\bibliography{abbrev0,crypto,bib}

\newcommand{\etalchar}[1]{$^{#1}$}
\begin{thebibliography}{KMCVY22}

\bibitem[AAB{\etalchar{+}}19]{supremacy}
Frank Arute, Kunal Arya, Ryan Babbush, Dave Bacon, Joseph~C Bardin, Rami Barends, Rupak Biswas, Sergio Boixo, Fernando~GSL Brandao, David~A Buell, et~al.
\newblock Quantum supremacy using a programmable superconducting processor.
\newblock {\em Nature}, 574(7779):505--510, 2019.

\bibitem[AABA{\etalchar{+}}24]{willow}
Rajeev Acharya, Laleh Aghababaie-Beni, Igor Aleiner, Trond~I Andersen, Markus Ansmann, Frank Arute, Kunal Arya, Abraham Asfaw, Nikita Astrakhantsev, Juan Atalaya, et~al.
\newblock Quantum error correction below the surface code threshold.
\newblock {\em arXiv preprint arXiv:2408.13687}, 2024.

\bibitem[AMR22]{TCC:AlaMalRah22}
Navid Alamati, Giulio Malavolta, and Ahmadreza Rahimi.
\newblock Candidate trapdoor claw-free functions from group actions with applications to quantum protocols.
\newblock In Eike Kiltz and Vinod Vaikuntanathan, editors, {\em TCC~2022: 20th Theory of Cryptography Conference, Part~I}, volume 13747 of {\em Lecture Notes in Computer Science}, pages 266--293, Chicago, IL, USA, November~7--10, 2022. Springer, Cham, Switzerland.

\bibitem[Bar89]{Barrington}
David~A. Barrington.
\newblock Bounded-width polynomial-size branching programs recognize exactly those languages in nc1.
\newblock {\em Journal of Computer and System Sciences}, 38(1):150--164, 1989.

\bibitem[BBK22]{BBK22}
Nir Bitansky, Zvika Brakerski, and Yael~Tauman Kalai.
\newblock Constructive post-quantum reductions.
\newblock In Yevgeniy Dodis and Thomas Shrimpton, editors, {\em Advances in Cryptology - {CRYPTO} 2022 - 42nd Annual International Cryptology Conference, {CRYPTO} 2022, Santa Barbara, CA, USA, August 15-18, 2022, Proceedings, Part {III}}, volume 13509 of {\em Lecture Notes in Computer Science}, pages 654--683. Springer, 2022.

\bibitem[BCM{\etalchar{+}}18]{FOCS:BCMVV18}
Zvika Brakerski, Paul Christiano, Urmila Mahadev, Umesh~V. Vazirani, and Thomas Vidick.
\newblock A cryptographic test of quantumness and certifiable randomness from a single quantum device.
\newblock In Mikkel Thorup, editor, {\em 59th Annual Symposium on Foundations of Computer Science}, pages 320--331, Paris, France, October~7--9, 2018. {IEEE} Computer Society Press.

\bibitem[BCM{\etalchar{+}}21]{BCMVV21}
Zvika Brakerski, Paul Christiano, Urmila Mahadev, Umesh Vazirani, and Thomas Vidick.
\newblock A cryptographic test of quantumness and certifiable randomness from a single quantum device.
\newblock {\em J. ACM}, 68(5), August 2021.

\bibitem[BGK{\etalchar{+}}23]{BGKPV23}
Zvika Brakerski, Alexandru Gheorghiu, Gregory~D. Kahanamoku{-}Meyer, Eitan Porat, and Thomas Vidick.
\newblock Simple tests of quantumness also certify qubits.
\newblock In Helena Handschuh and Anna Lysyanskaya, editors, {\em Advances in Cryptology - {CRYPTO} 2023 - 43rd Annual International Cryptology Conference, {CRYPTO} 2023, Santa Barbara, CA, USA, August 20-24, 2023, Proceedings, Part {V}}, volume 14085 of {\em Lecture Notes in Computer Science}, pages 162--191. Springer, 2023.

\bibitem[BK24]{BK24}
James Bartusek and Dakshita Khurana.
\newblock On the power of oblivious state preparation.
\newblock {\em CoRR}, abs/2411.04234, 2024.

\bibitem[BKM{\etalchar{+}}24]{BKMSW24}
Kaniuar Bacho, Alexander Kulpe, Giulio Malavolta, Simon Schmidt, and Michael Walter.
\newblock Compiled nonlocal games from any trapdoor claw-free function.
\newblock Cryptology {ePrint} Archive, Paper 2024/1829, 2024.

\bibitem[BKW97]{BKW97}
Johannes Bl\"{o}mer, Richard Karp, and Emo Welzl.
\newblock The rank of sparse random matrices over finite fields.
\newblock {\em Random Struct. Algorithms}, 10(4):407–419, July 1997.

\bibitem[CCM98]{CCM98}
C.~Cachin, C.~Crepeau, and J.~Marcil.
\newblock Oblivious transfer with a memory-bounded receiver.
\newblock In {\em Proceedings 39th Annual Symposium on Foundations of Computer Science (Cat. No.98CB36280)}, pages 493--502, 1998.

\bibitem[CM97]{CM97}
Christian Cachin and Ueli~M. Maurer.
\newblock Unconditional security against memory-bounded adversaries.
\newblock In Burton S.~Kaliski Jr., editor, {\em Advances in Cryptology - {CRYPTO} '97, 17th Annual International Cryptology Conference, Santa Barbara, California, USA, August 17-21, 1997, Proceedings}, volume 1294 of {\em Lecture Notes in Computer Science}, pages 292--306. Springer, 1997.

\bibitem[DHRS04]{DHRS04}
Yan~Zong Ding, Danny Harnik, Alon Rosen, and Ronen Shaltiel.
\newblock Constant-round oblivious transfer in the bounded storage model.
\newblock In Moni Naor, editor, {\em Theory of Cryptography, First Theory of Cryptography Conference, {TCC} 2004, Cambridge, MA, USA, February 19-21, 2004, Proceedings}, volume 2951 of {\em Lecture Notes in Computer Science}, pages 446--472. Springer, 2004.

\bibitem[Din01]{Ding01}
Yan~Zong Ding.
\newblock Oblivious transfer in the bounded storage model.
\newblock In Joe Kilian, editor, {\em Advances in Cryptology --- CRYPTO 2001}, pages 155--170, Berlin, Heidelberg, 2001. Springer Berlin Heidelberg.

\bibitem[DORS06]{DORS08}
Yevgeniy Dodis, Rafail Ostrovsky, Leonid Reyzin, and Adam~D. Smith.
\newblock Fuzzy extractors: How to generate strong keys from biometrics and other noisy data.
\newblock {\em CoRR}, abs/cs/0602007, 2006.

\bibitem[DQW23]{DQW23}
Yevgeniy Dodis, Willy Quach, and Daniel Wichs.
\newblock Speak much, remember little: Cryptography in the bounded storage model, revisited.
\newblock In Carmit Hazay and Martijn Stam, editors, {\em Advances in Cryptology - {EUROCRYPT} 2023 - 42nd Annual International Conference on the Theory and Applications of Cryptographic Techniques, Lyon, France, April 23-27, 2023, Proceedings, Part {I}}, volume 14004 of {\em Lecture Notes in Computer Science}, pages 86--116. Springer, 2023.

\bibitem[GL89]{GL89}
O.~Goldreich and L.~A. Levin.
\newblock A hard-core predicate for all one-way functions.
\newblock In {\em Proceedings of the Twenty-First Annual ACM Symposium on Theory of Computing}, STOC '89, page 25–32, New York, NY, USA, 1989. Association for Computing Machinery.

\bibitem[GZ19]{EC:GuaZha19}
Jiaxin Guan and Mark Zhandry.
\newblock Simple schemes in the bounded storage model.
\newblock In Yuval Ishai and Vincent Rijmen, editors, {\em Advances in Cryptology -- {EUROCRYPT}~2019, Part~III}, volume 11478 of {\em Lecture Notes in Computer Science}, pages 500--524, Darmstadt, Germany, May~19--23, 2019. Springer, Cham, Switzerland.

\bibitem[KCVY22]{KCVY22}
Gregory Kahanamoku\text{-}Meyer, Soonwon Choi, Umesh Vazirani, and Norman Yao.
\newblock Classically verifiable quantum advantage from a computational bell test.
\newblock {\em Nature Physics}, 18:1--7, 08 2022.

\bibitem[KLVY23a]{STOC:KLVY23}
Yael Kalai, Alex Lombardi, Vinod Vaikuntanathan, and Lisa Yang.
\newblock Quantum advantage from any non-local game.
\newblock In Barna Saha and Rocco~A. Servedio, editors, {\em 55th Annual {ACM} Symposium on Theory of Computing}, pages 1617--1628, Orlando, FL, USA, June~20--23, 2023. {ACM} Press.

\bibitem[KLVY23b]{KLVY23}
Yael Kalai, Alex Lombardi, Vinod Vaikuntanathan, and Lisa Yang.
\newblock Quantum advantage from any non-local game.
\newblock In Barna Saha and Rocco~A. Servedio, editors, {\em Proceedings of the 55th Annual {ACM} Symposium on Theory of Computing, {STOC} 2023, Orlando, FL, USA, June 20-23, 2023}, pages 1617--1628. {ACM}, 2023.

\bibitem[KMCVY22]{dlog}
Gregory~D Kahanamoku-Meyer, Soonwon Choi, Umesh~V Vazirani, and Norman~Y Yao.
\newblock Classically verifiable quantum advantage from a computational bell test.
\newblock {\em Nature Physics}, 18(8):918--924, 2022.

\bibitem[LLL{\etalchar{+}}21]{closinggap}
Yong Liu, Xin Liu, Fang Li, Haohuan Fu, Yuling Yang, Jiawei Song, Pengpeng Zhao, Zhen Wang, Dajia Peng, Huarong Chen, et~al.
\newblock Closing the" quantum supremacy" gap: achieving real-time simulation of a random quantum circuit using a new sunway supercomputer.
\newblock In {\em Proceedings of the International Conference for High Performance Computing, Networking, Storage and Analysis}, pages 1--12, 2021.

\bibitem[LZG{\etalchar{+}}24]{experimental}
Laura Lewis, Daiwei Zhu, Alexandru Gheorghiu, Crystal Noel, Or~Katz, Bahaa Harraz, Qingfeng Wang, Andrew Risinger, Lei Feng, Debopriyo Biswas, et~al.
\newblock Experimental implementation of an efficient test of quantumness.
\newblock {\em Physical Review A}, 109(1):012610, 2024.

\bibitem[Mah18a]{Mah18FHE}
Urmila Mahadev.
\newblock Classical homomorphic encryption for quantum circuits.
\newblock In {\em 2018 IEEE 59th Annual Symposium on Foundations of Computer Science (FOCS)}, pages 332--338, 2018.

\bibitem[Mah18b]{CVQC}
Urmila Mahadev.
\newblock Classical verification of quantum computations.
\newblock In {\em 2018 IEEE 59th Annual Symposium on Foundations of Computer Science (FOCS)}, pages 259--267. IEEE, 2018.

\bibitem[Mau92]{Maurer92}
Ueli~M. Maurer.
\newblock Conditionally-perfect secrecy and a provably-secure randomized cipher.
\newblock {\em J. Cryptol.}, 5(1):53--66, 1992.

\bibitem[MSY24]{MSY24}
Tomoyuki Morimae, Yuki Shirakawa, and Takashi Yamakawa.
\newblock Cryptographic characterization of quantum advantage.
\newblock Cryptology {ePrint} Archive, Paper 2024/1536, 2024.

\bibitem[MY23]{MY23}
Tomoyuki Morimae and Takashi Yamakawa.
\newblock Proofs of quantumness from trapdoor permutations.
\newblock In Yael~Tauman Kalai, editor, {\em 14th Innovations in Theoretical Computer Science Conference, {ITCS} 2023, January 10-13, 2023, MIT, Cambridge, Massachusetts, {USA}}, volume 251 of {\em LIPIcs}, pages 87:1--87:14. Schloss Dagstuhl - Leibniz-Zentrum f{\"{u}}r Informatik, 2023.

\bibitem[NOVY92]{NOVY}
Moni Naor, Rafail Ostrovsky, Ramarathnam Venkatesan, and Moti Yung.
\newblock Perfect zero-knowledge arguments for {NP} can be based on general complexity assumptions (extended abstract).
\newblock In Ernest~F. Brickell, editor, {\em Advances in Cryptology - {CRYPTO} '92, 12th Annual International Cryptology Conference, Santa Barbara, California, USA, August 16-20, 1992, Proceedings}, volume 740 of {\em Lecture Notes in Computer Science}, pages 196--214. Springer, 1992.

\bibitem[NZ23]{NZ23}
Anand Natarajan and Tina Zhang.
\newblock Bounding the quantum value of compiled nonlocal games: From {CHSH} to {BQP} verification.
\newblock In {\em 64th {IEEE} Annual Symposium on Foundations of Computer Science, {FOCS} 2023, Santa Cruz, CA, USA, November 6-9, 2023}, pages 1342--1348. {IEEE}, 2023.

\bibitem[PS97]{PS97}
Alessandro Panconesi and Aravind Srinivasan.
\newblock Randomized distributed edge coloring via an extension of the chernoff--hoeffding bounds.
\newblock {\em SIAM Journal on Computing}, 26(2):350--368, 1997.

\bibitem[Raz18]{Raz}
Ran Raz.
\newblock Fast learning requires good memory: A time-space lower bound for parity learning.
\newblock {\em J. ACM}, 66(1), December 2018.

\bibitem[Sho94]{FOCS:Shor94}
Peter~W. Shor.
\newblock Algorithms for quantum computation: Discrete logarithms and factoring.
\newblock In {\em 35th Annual Symposium on Foundations of Computer Science}, pages 124--134, Santa Fe, NM, USA, November~20--22, 1994. {IEEE} Computer Society Press.

\bibitem[Top01]{Topsoe01}
Flemming Topsøe.
\newblock Bounds for entropy and divergence for distributions over a two-element set.
\newblock {\em JIPAM. Journal of Inequalities in Pure \& Applied Mathematics [electronic only]}, 2(2):Paper No. 25, 13 p.--Paper No. 25, 13 p., 2001.

\bibitem[Vad04]{JC:Vadhan04}
Salil~P. Vadhan.
\newblock Constructing locally computable extractors and cryptosystems in the bounded-storage model.
\newblock {\em Journal of Cryptology}, 17(1):43--77, January 2004.

\end{thebibliography}

\appendix
\iflncs
\section{Analysis of \Cref{protocol:template}}\label{sec:appendix}

\else
\section{Classical Verification of BQP}\label{sec:appendix}

It is shown in~\cite{NZ23,BK24,BKMSW24} that a claw-generation protocol with quantum soundness implies the existence of a protocol where a quantum prover can demonstrate to a completely classical verifier the validity of any statement in BQP. The protocol roughly goes as follows: First, it turns the hardness of finding a claw into a protocol for \emph{blind quantum computation}, i.e., a protocol where the circuit that is computed by the quantum prover is computationally hidden, and the verifier is fully classical. Second, it uses blind quantum computation to \emph{compile} a two-player non-local game for BQP verification~\cite{NZ23,KLVY23} into a single-player one. %The analysis of~\cite{NZ23} yields the desired result. \tamer{the last sentence is confusing given the beginning ``it is shown in BK24''}

%We expect that, b
By plugging our claw-generation protocol into the framework of~\cite{BK24}, one can obtain a protocol for general verification of BQP computation against \emph{memory-bounded} quantum adversaries, with \emph{unconditional} security. However, the straightforward combination of \cref{lem:claw-gen} with the~\cite{BK24} theorem meets two main discrepancies.

First, the reductions as stated in~\cite{BK24} are not concerned with preserving the memory complexity of the attacker, which on the other hand is necessary in our setting. In \cref{sec:appendix-memory-bounded}, we outline how to adapt all reductions involved to be ``memory preserving''.

%: The first step is to turn a claw-generation protocol into an \emph{oblivious state preparation} (OSP) protocol, where a classical verifier instructs a prover to prepare either a Hadamard- or a computational-basis state, without revealing which one. The proof of this step consists of an invocation of the quantum Goldreich-Levin theorem \cite{QGL}, where the extractor consists of a single query to the adversary's unitary (and its inverse) and uses a single additional ancilla. Thus, this step respects the memory bound of the attacker. The second step builds a blind evaluation of a quantum circuit starting from an OSP. Since the blindness comes exclusively from the OSP protocol, the argument is straightforward.

Second, the analysis from~\cite{BK24} requires a claw-generation protocol where advantage of claw-finding is negligible, whereas \cref{lmm:qhard} only provides an inverse-polynomial bound on the success probability of the attacker. We show how to bridge this gap in \cref{sec:appendix-soundness-error}.%, we explain why the soundness guarantee of \cref{lmm:qhard} is still strong enough to allow for the transformation 

\subsection{From Claw-Generation to Verification of BQP}\label{sec:appendix-memory-bounded}

We outline how to construct a classical verification for BQP computations, unconditionally secure against memory-bounded adversaries. In what follows, we assume that we have a claw-generation protocol with negligible soundness error (the success probability of any memory-bounded attacker) and obtain a classical verification protocol for BQP computations with inverse-polynomial gap between completeness (the success probability of an honest prover) and soundness error.% We relax this assumption in \cref{sec:appendix-soundness-error}. 

In a nutshell, we follow the strategy from \cite{BK24} where an analogous implication is shown in the standard cryptographic setting (against computationally-bounded attackers). In fact, all the steps are precisely identical, except that we need to argue why security holds against memory-bounded adversaries. Given that the protocol and the arguments are unchanged, we only provide a proof sketch.

\paragraph{Step I: Oblivious State Preparation.} An oblivious state preparation (OSP) \cite{BK24} is a protocol between a classical verifier and a quantum prover. At the end of the interaction, the prover holds the state
\[
H^\theta \ket{b}
\]
whereas the verifier holds the bits $(\theta, b)$. Security requires that any QPT prover cannot guess the bit of the verifier $\theta$ with probability non-negligibly greater than $1/2$. It is shown in \cite{BK24} (Theorem 4.7) that a claw-state generation protocol can be generically used to construct an OSP. First it is shown that a claw-state generation protocol can always be assumed without loss of generality (Lemma 4.6 in \cite{BK24}) to prepare a state of the form
\[
\frac{\ket{0, x_0} + \ket{1, x_1}}{\sqrt{2}} 
\]
for $x_0,x_1\in\{0,1\}^n$.
Then the prover and the verifier engage in the following interactive protocol:
\begin{itemize}
    \item The verifier samples two random bitstrings $r_0,r_1 \gets \{0,1\}^n$.
    \item The prover applies the map
    \[
    \frac{\ket{0, x_0} + \ket{1, x_1}}{\sqrt{2}}  \mapsto \frac{\ket{0, x_0, x_0^\intercal r_0} + \ket{1, x_1, x_1^\intercal r_1}}{\sqrt{2}} 
    \]
    and measures all but the last qubit in the Hadamard basis to obtain a string $d\in\{0,1\}^{n+1}$.
    \item The verifier sets $\theta = (x_0, x_1)^\intercal (r_0, r_1)$. If $\theta = 0$, then it sets $b = x_0^\intercal r_0 = x_1^\intercal r_1$. Otherwise, it sets $b = d ^\intercal (1,x_0 \oplus x_1)$.
\end{itemize}
Correctness follows by direct calculation, whereas security follows by reducing the hardness of guessing $\theta$ to the hardness of computing $(x_0, x_1)$, with the Goldreich-Levin search-to-decision reduction. Using \cref{lmm:GL}, the same reduction holds in the memory-bounded settings.

\paragraph{Step II: Blind Delegation.} A blind delegation protocol allows a verifier to delegate the computation of a quantum circuit $Q$ on a classical input $\ket{y}$, while keeping $y$ hidden. It is shown in \cite{BK24} (Theorem 6.11) that OSP implies a blind delegation protocol. The protocol starts with a one-time padded input $X^x\ket{y}$ and proceeds while the verifier keeps track of the one-time pad keys and the prover performs the quantum gates. For a Clifford gates $C$, the following identity is used:
\[
C X^x Z^z = X^{x'} Z^{z'} C
\]
where $x',z'$ are functions of $x,z$. For non-Clifford gates, \cite{BK24} (Theorem 6.10) shows an interactive protocol to correct the errors, based on OSP. This subroutine, referred to as \emph{encrypted phase}, allows a prover to obliviously apply a phase, conditioned on a bit that is only known to the verifier. The protocol proceeds as follows:
\begin{itemize}
    \item The verifier and the prover engage in an OSP protocol, with the verifier's bit  $(\theta, b)$. The prover applies a Hadamard gate $H$ and a $\sqrt{X}$ gate to its state. This results into the state
    \[
    Z^b P^\theta \ket{+}.
    \]
    \item The prover CNOTs the target register onto the above state, and measures it in the computational basis to obtain a bit $m\in\{0,1\}$.
    \item The prover sents $m$ to the verifier, who applies the appropriate correction to the one-time pad.
\end{itemize}
We omit most details from this outline, but what is important for us is that the blindness follows directly from the indistinguishability property of the OSP. Thus, precisely the same analysis holds in the memory-bounded settings.

\paragraph{Step III: Computationally Non-Local Strategies.} The last step in the strategy is to use the  blind-delegation protocol to \emph{compile} a two-prover non-local game into a single prover one. This idea was first proposed in \cite{KLVY23} and it soundness was analyzed in \cite{NZ23}, for the case of a particular blind-delegation protocol, based on quantum fully-homomorphic encryption. In \cite{BK24} (Theorem 6.15) this approach is extended to \emph{any} blind delegation protocol, and it shown that the compilation results into a single-prover \emph{computationally non-local strategy}. Blindness of the delegation protocol is only used to establish that the maximal success probability of any QPT prover in the compiled protocol is bounded by the supremum over all computationally non-local strategies (Theorem 6.23). This is again a direct reduction to the blind delegation protocol that also works in the memory-bounded settings (provided of course that the prover is memory bounded). All subsequent steps of the analysis are information-theoretic and thus trivially hold for our setting as well.

\subsection{Tolerating Non-Negligible Soundness Error}\label{sec:appendix-soundness-error}

Next, we show that our claw-generation protocol (\cref{lem:claw-gen}) implies classical verification of BQP computation, despite its inverse-polynomial soundness error against quantum attackers, thus relaxing the assumption of negligible soundness error made in \cref{sec:appendix-memory-bounded}.

One simple way to derive this is by relying on the fact that our quantum soundness argument from \cref{lmm:qhard} stems from an indistinguishability argument between the real experiment and an ideal experiment, where soundness holds \emph{statistically}. 

Let us recall some notation from the proof of \cref{lmm:qhard}: $Z$ is the random variable containing the ``hard'' part of the claw consisting of bits $\{z^1_b,\dots,z^\sec_b\}_{b\in\bin}$ and $W$ is the adversary's $m$-qubit state at the end of the protocol. \cref{lmm:qhard} bounds that trace-distance between $(J,Z_J,W)$ and $(J,Z',W)$, where $J$ is a subset of size at least $\sec/2$ and $Z'$ is uniformly random, by $\epsilon = \sec^2\sqrt{m/2k}$. Since the polynomial $k$ can be chosen to be arbitrarily larger than the bound on memory $m$, we can choose $\epsilon$ to be an arbitrarily small inverse-polynomial function of $\sec$. Specifically, let $\delta$ denote the inverse-polynomial gap between completeness and soundness error in the protocol from~\cite{BK24} when instantiated using claw-generation with negligible soundness error, which we inherit in the bounded-memory setting as shown in \cref{sec:appendix-memory-bounded}. Let $T=\poly(\sec)$ denote the number of times claw-generation is invoked in the protocol to prove a given BQP instance. We set $\epsilon$ to be small enough so that $T\cdot \epsilon< \delta/2$. We argue that, when the protocol is instantiated using our claw-generation from \cref{lem:claw-gen}, we obtain gap between completeness and soundness error that is at most $\delta/2$.

To see this, consider an ideal experiment where after any invocation of claw-generation throughout the protocol, the verifier modifies the claw values in its memory by replacing the bits at locations $J$ in $Z$ with freshly sampled uniformly random bits. By \cref{lmm:qhard}, any such replacement deviates us from the real experiment by at most $\epsilon$ in trace-distance. In total, we obtain an experiment that is at most $\delta/2$-far from the real experiment by triangle inequality.

In the ideal experiment, claw-generation provides soundness up to error $2^{-\sec/2}$, which is negligible in the security parameter. Following the adaption of \cite{BK24} to our setting (\cref{sec:appendix-memory-bounded}), this implies that an attacker against the protocol in the experiment succeeds with soundness error that is smaller than completeness by at least $\delta$. Therefore, this gap in the real experiment is at least $\delta/2$ by triangle inequality.

\fi

\end{document}